\documentclass{article}

\usepackage{arxiv}

\usepackage{verbatim}
\usepackage{amsmath, amsthm, amssymb, epsfig, sectsty, graphicx, float, url}
\usepackage{bm}
\usepackage{algorithm, algorithmicx}
\usepackage[noend]{algpseudocode}

\usepackage[utf8]{inputenc} % allow utf-8 input
\usepackage[T1]{fontenc}    % use 8-bit T1 fonts
\usepackage{hyperref}       % hyperlinks
\usepackage{url}            % simple URL typesetting
\usepackage{booktabs}       % professional-quality tables
\usepackage{amsfonts}       % blackboard math symbols
\usepackage{nicefrac}       % compact symbols for 1/2, etc.
\usepackage{microtype}      % microtypography
\usepackage{cleveref}

\theoremstyle{plain}
\newtheorem{theorem}{Theorem}[section]
\newtheorem{lemma}[theorem]{Lemma}

\newtheorem{corollary}[theorem]{Corollary}
\newtheorem{proposition}[theorem]{Proposition}

\newtheorem{assumption}[theorem]{Assumption}
\theoremstyle{definition}
\newtheorem{definition}[theorem]{Definition}
\newtheorem{remark}[theorem]{Remark}

\usepackage{mathtools}
\usepackage{amsmath, amsthm, amssymb, epsfig, sectsty, graphicx, float, url}
\usepackage{algorithm, algorithmicx, algpseudocode}
\usepackage{enumitem}
\usepackage[disable,colorinlistoftodos,bordercolor=orange,backgroundcolor=orange!20,linecolor=orange,textsize=scriptsize]{todonotes}
\catcode`\@=11
\catcode`\@=12

\newcommand{\opt}{^{\star}}

\def\<#1,#2>{\langle #1,#2\rangle_{\alpha}}
\def\p<#1,#2>{\langle #1,#2\rangle_{1}}
\def\q[#1]{[#1]_{\alpha}}
\def\1q[#1]{[#1]_{1}}

\newcommand{\Max}{\operatorname{Max}}
\newcommand{\Min}{\operatorname{Min}}
\newcommand{\tr}{\operatorname{tr}}
\newcommand{\cof}{\operatorname{cof}}
\newcommand{\adj}{\operatorname{adj}}
\newcommand{\comp}{\mathcal{C}}
\newcommand{\N}{\mathbb{N}}
\newcommand{\Z}{\mathbb{Z}}
\newcommand{\R}{\mathbb{R}}
\newcommand{\E}{\mathbb{E}}

\title{Thresholds for sensitive optimality and Blackwell optimality in stochastic games}

\author{St\'ephane Gaubert \\
INRIA and CMAP, Ecole Polytechnique
\\
stephane.gaubert@inria.fr \\
\And
Julien Grand-Cl\'ement\\
ISOM Department, HEC Paris\\grand-clement@hec.fr \\
\And
Ricardo D. Katz \\
CIFASIS-CONICET \\
katz@cifasis-conicet.gov.ar\\
}

\newcommand{\sens}{d}

\newcommand{\alphabw}{\alpha_{\sf Bw}}

\newcommand{\alphama}{\alpha_{\sf Ma}}
\newcommand{\alphad}{\alpha_{\sf d}}

\newcommand{\de}{\mathrm{det}}
\newcommand{\sto}{\mathrm{sto}}

\begin{document}

\maketitle
\vspace{4mm}

\begin{abstract}
We investigate refinements of the mean-payoff criterion in two-player zero-sum perfect-information stochastic games.
  A strategy is {\em Blackwell optimal} if it is optimal in the discounted game for all discount factors sufficiently close to $1$. The notion of {\em $d$-sensitive optimality} interpolates between mean-payoff optimality (corresponding to the case $\sens=-1$) and Blackwell optimality ($\sens=\infty$). The {\em Blackwell threshold} $\alphabw \in [0,1[$ is the discount factor above which all 
  optimal strategies in the discounted game
  are guaranteed to be Blackwell optimal. 
  The {\em $\sens$-sensitive threshold} $\alphad \in [0,1[$ is defined analogously. 
      Bounding $\alphabw$ and $\alphad$ are fundamental problems in algorithmic game theory, since these thresholds control the complexity for computing Blackwell and $\sens$-sensitive optimal strategies, by reduction to discounted games which can be solved in $O\left((1-\alpha)^{-1}\right)$ iterations.
      %% We derive new improved bounds for $\alphabw$ and $\alphad$ for stochastic games.
      We provide the first bounds on the $\sens$-sensitive threshold $\alphad$ beyond the case $d=-1$, and we establish improved bounds for the Blackwell threshold $\alphabw$. This is achieved by leveraging separation bounds on algebraic numbers, relying on Lagrange bounds and more advanced techniques based on Mahler measures and multiplicity theorems.
\end{abstract}

%%%%% Commenting the next section so that we have an accurate count for the number of pages
% \section*{TODOs/to discuss}
% {\bf to do}
% \begin{enumerate}
% \item add short lit. review on algorithms/complexity for the problem we study
% \item fill out Table \ref{tab:bound deterministic SGs}, Table \ref{tab:bound general SGs} (missing some $O(\cdot)$ bounds on $\alphabw,\alphad$)
% \item Missing tables of different regimes in the stochastic case (analogous to Table \ref{table-regimes}).
% \item Questions/to discuss:
% \begin{enumerate}
% \item Does Theorem \ref{th:bound on k} (results using multiplicity) improve upon Theorem \ref{ThLagrange} and Theorem \ref{ThMahler}? If not we should move this as a remark in later sections (currently in the intro).
% \item Do we want to add a table with the corresponding bounds we obtain by plugging $\alphabw$ in the complexity bound on strategy iterations? Maybe in the last very last section? But we're already out of space.
% \end{enumerate}
% \item Mention complexity of computing Blackwell optimal policy in literature review
% \end{enumerate}

\section{Introduction}\label{sec:intro}
Two-player perfect-information zero-sum stochastic games (SGs) are an important class of Shapley's stochastic games \cite{shapley1953stochastic} where every state is controlled by a unique player.
SGs are a cornerstone of game theory, with important applications in auctions~\cite{lazarus1999combinatorial}, mechanism design~\cite{narahari2014game}, multi-agent reinforcement learning~\cite{littman1994markov}, robust optimization~\cite{grand2023beyond,chatterjee2023solving}, and $\mu$-calculus model-checking~\cite{clarke2018handbook}.
%tug of war games~\cite{peres2009tug} 
Shapley originally considered the sum of discounted instantaneous rewards as the objective; the mean-payoff objective was studied in~\cite{gillette1957stochastic}.
Stationary and deterministic optimal strategies exist for these objectives
in the perfect-information case~\cite{shapley1953stochastic,liggettlippman}.

The mean-payoff objective coincides with the limit of the discounted objective as the discount factor goes to $1$. This suggests to study the variation of the set of optimal strategies in the discounted game, as a function of the discount factor.
Blackwell showed (in the one-player case) that, for finite state and action models, there exist strategies that remain discount optimal for all values of the discount factor that are sufficiently close to $1$~\cite{blackwell1962discrete}. These are now called {\em Blackwell optimal strategies}.
As the discount factor approaches $1$, the ``weight'' given to the rewards received in the distant future increases. Therefore, Blackwell optimal strategies can be considered as the most farsighted (or least greedy) strategies for models in long horizon. This is why better understanding Blackwell optimality is referred to as ``{\em one of the pressing questions in reinforcement learning}'' in~\cite{dewanto2020average}. Veinott introduced (still in the one-player case)
the notion of {\em $\sens$-sensitive optimality}, interpolating between mean-payoff optimality (obtained for $\sens=-1$) and Blackwell optimality (obtained
for $\sens=\infty$). The value associated with a given stationary strategy has a Laurent series expansion in the powers of $1-\alpha$, where $\alpha$ is the discount factor, with a pole of order at most $-1$ at $\alpha= 1$. Veinott's $\sens$-sensitive objective involves the sequence of coefficients of this expansion, up to order $\sens$, which is maximized or minimized with respect to the lexicographic order (see Theorem 10.1.6 in~\cite{puterman2014markov}). In particular, any $\sens$-sensitive optimal strategy is $\sens'$-sensitive
optimal for all $\sens'<\sens$. Moreover, for models with $n$ states, every $(n-2)$-sensitive optimal strategy is Blackwell optimal. The notion for $\sens=0$ is also known as {\em bias optimality}.

It is worth emphasizing that the notions of Blackwell optimality and mean-payoff optimality have received increased attention recently in the reinforcement learning community, see~\cite{dewanto2020average,dewanto2022examining,tang2021taylor,yang2016efficient}. More generally, computing mean-payoff optimal strategies is an important open question in algorithmic game theory and has been extensively studied. {\em Pseudo-}polynomial algorithms exist for mean-payoff {\em deterministic} instances, based on {\em pumping}~\cite{gurvich1988cyclic} and value iteration~\cite{zwick,cuninghame2003equation,gaubert2008cyclic}. In more generality, computing mean-payoff or bias optimal strategies is difficult, since the Bellman operator is no longer a contraction when $\alpha=1$, and the mean-payoff and bias objectives may be discontinuous in the entries of stationary strategies (e.g., Chapter~4 of~\cite{feinberg2012handbook}). In fact, strategy iteration may cycle for the mean-payoff objective in multichain instances (see Section~6 of \cite{akian2012policy}), and other algorithms have an exponential dependence on the number of states with stochastic (non-deterministic) transitions or undetermined complexities~\cite{akian2012policy,boros2010pumping}. 
% \jgc{is this true for \cite{jurdzinski2008deterministic}?}
\todo{more recent/classical references?}
Besides the method based on the Blackwell
threshold discussed below, the only algorithm we are aware of to compute
Blackwell optimal policies consists in performing
policy iteration, considering the discount factor
as a formal parameter, and encoding the discounted value function
associated to a pair of policies by its Laurent series expansion
truncated at order $n-2$, leading to a nested lexicographic
policy iteration method, see~\cite[Chapter 10.3]{puterman2014markov}
for a presentation in the one player case.
A counter example of Friedmann implies that for deterministic games,
this algorithm can take an exponential time~\cite{Friedmann-AnExponentialLowerB}.

% \jgc{Check Chaloupka Cha09 paper for comparisons?}

%% Traditional notions of optimality for perfect-information SGs include {\em \sens-sensitive optimality}~\cite{Veinott1969}, where the players optimize for $\sens +2$ nested excess reward functions obtained over an infinite horizon ($\sens =-1$ corresponds to {\em mean-payoff} optimality~\cite{gillette1957stochastic} and $d=0$ corresponds to {\em bias} optimality), {\em discount optimality}, where the players optimize the discounted sum of the rewards (in which the discount factor $\alpha \in [0 , 1[$  models the diminishing importance of future rewards), and {\em Blackwell optimality}, which corresponds to strategies that remain discount optimal for all discount factors $\alpha$ sufficiently close to $1$.  In fact, Blackwell optimal strategies are also mean-payoff optimal, and $\sens$-sensitive optimal strategies are Blackwell optimal for $\sens=n-2$ and mean-payoff optimal for $\sens=-1$. Additionally, for sufficiently large discount factors, discount optimal strategies are also mean-payoff optimal {\em and} Blackwell optimal. 
%% % \todoi{was it shown somewhere that $\alphad$ exists?}

The existence of the {\em Blackwell threshold}, defined as the smallest number $\alphabw  \in [0 , 1[$ such that strategies that are optimal for any discount factor $\alpha\geq \alphabw$ are Blackwell optimal, is well-known, e.g.~\cite{andersson2009complexity}. We define the {\em $\sens$-sensitive threshold} as the smallest number $\alphad \in [0 , 1[$ such that strategies that are discount optimal for any discount factor $\alpha\geq \alphad$ are $d$-sensitive optimal. The existence of $\alphabw$ and $\alphad$ suggests a simple method for computing Blackwell and $\sens$-sensitive optimal strategies: compute discount optimal strategies for a discount factor $\alpha$ close to 1. This approach has the advantage that {\em any} advances in solving discounted SGs transfer to algorithms for Blackwell and mean-payoff optimal strategies.
Note that the term $(1-\alpha)^{-1}$ controls the complexity of computing discount optimal strategies: for known game parameters (rewards and transitions), strategy iteration and value iteration scale as $\tilde{O}((1-\alpha)^{-1})$~\cite{hansen2013strategy} (hiding the dependence on the number of states/actions, and lower order terms in $1-\alpha$). For unknown game parameters, model-based methods~\cite{zhang2023model} or model-free methods based on Q-learning~\cite{sidford2020solving} scale as $\tilde{O}((1-\alpha)^{-3})$.
% Indeed, discount optimal strategies can be computed using classical algorithms like strategy iteration~\cite{hansen2013strategy} or a combination of strategy iteration with combinatorial interior point methods~\cite{ye2005new}. This last method returns a discount optimal strategy in $O((1-\alpha)^{-1}\log((1-\alpha)^{-1}))$ arithmetic operations (the dependence on the number of states and actions is polynomial and hidden in the $O(.)$ notation). 

Thus, for computational purposes, it is crucial to upper bound $\alphabw$ and $\alphad$, or equivalently to lower bound $1-\alphabw$ and $1-\alphad$. 
% For such bounds to be useful, it is important that they rely solely on a few parameters of the SG that are available by design, such as the number of states $n \in \N$, the maximum absolute value $W \in \N$ of the rewards (assuming the rewards are integer-valued), or the common denominator $M \in \N$ in the transition probabilities governing the dynamics of the game.
For such bounds to be useful, they should only rely on a few game parameters that are available by design. In this paper we consider SGs that satisfy the next assumption.
\begin{assumption}\label{assumption}
$\Gamma$ is a perfect-information stochastic game with $n \in \N$ states, integer rewards with absolute values bounded by $W \in \N$, and transition probabilities with common denominator $M \in \N$.
% is the common denominator appearing in the transition probabilities.
\end{assumption}

% which we such as the number of states $n \in \N$, the maximum absolute value $W \in \N$ of the rewards (assuming the rewards are integer-valued), or the common denominator $M \in \N$ in the transition probabilities governing the dynamics of the game.

\noindent
{\bf Main results.}
We provide upper bounds for the Blackwell threshold $\alphabw$ and for the $d$-sensitive threshold $\alphad$ for perfect-information SGs.
We also derive stronger results in the case of deterministic games (transition probabilities in $\{0,1\}$) or for stochastic games
with unichain structure.
%When the game parameters (payments and transition probabilities) are fixed, these thresholds are algebraic numbers. 
%At the core of our proofs lie techniques for {\em separating} $\alphabw$ and $\alphad$ from $1$, i.e., to provide a lower bound on $|\alphabw-1|$ given the game parameters.

At the core of our results are separation bounds for algebraic numbers, a classical topic in algebraic number theory. %We construct a polynomial $\Delta$ such that $\Delta(\alphabw)=0$ (or $\Delta(\alphad)=0$ if studying $\alphad$). We then rely on three separation methods. 
 More precisely, we use three separation methods. The first method relies on a result of Lagrange~\cite{Lagrange}, also obtained by Hadamard~\cite{hadamard1893} and Fujiwara~\cite{fujiwara1916obere}, providing a lower bound on the
modulus $|z|$ of any non-zero root $z$ of a polynomial $P$, see Lecture IV in~\cite{yap2000fundamental}. %; we can apply this bound to the polynomial $\alpha \mapsto \Delta(1-\alpha)$ to lower bound $|1-\alphabw|$.
The second method relies on {\em Mahler measures} to lower bound $|1-z|$ when $z \neq 1$ is a root of a polynomial $P$, especially on a theorem of Dubickas~\cite{Dubickas1995}  building on a series of earlier works based on the seminal paper by Mignotte and Waldshmidt~\cite{mignotte1994algebraic}. %We apply this bound directly to $\Delta$. 
The third method uses a bound of Borwein, Ed\'erlyi, and K\'os on the multiplicity of $1$ as a root of a polynomial with integer coefficients~\cite{BEK1999}, which controls the value of $d$ such that $\alphad=\alphabw$.
We shall see that each of these approaches yields a useful bound -- not dominated by the other bounds in some regimes of the game parameters $n$, $W$ and $M$.
  %% To do so, we use results due to Lagrange to separate algebraic numbers from $0$~\cite{yap2000fundamental,fujiwara1916obere}, results based on Mahler measures to separate algebraic numbers from $1$~\cite{mignotte1994algebraic,Dubickas1995}, and results on the multiplicity of $1$ as a root of a polynomial~\cite{BEK1999}.
We present our main results below, according to the different techniques. For the sake of readability, we simplify some of our results with $O(.)$ notations here, and we defer the detailed statements to the next sections. We start with the results based on Lagrange.
\begin{theorem}[Based on Lagrange bound]\label{ThLagrange}
%Assume that $\Gamma$ is a perfect-information stochastic game with $n \in \N$ states, integer rewards with absolute values bounded by $W \in \N$, and $M \in \N$ is the common denominator appearing in the transition probabilities.
If the game $\Gamma$ satisfies \Cref{assumption}, then:
 
\noindent
{\em Deterministic case ($M=1$).} The \sens-sensitive threshold $\alphad$ and the Blackwell threshold $\alphabw$ satisfy
  \begin{equation}\label{eq: bound $\sens$-sensitive + Blackwell discount factor - deterministic}
  \alphad \leq 1-\frac{1}{24 W{2n\choose {\min \{ \sens+4 , n\}}}} \;  \makebox{ and }\; \alphabw \leq 1-\frac{1}{24 W{2n\choose n}} \; .
  \end{equation}
{\em Stochastic case ($M>1$).} The Blackwell threshold $\alphabw$ satisfies
$\displaystyle\alphabw \leq 
1-\frac{2^{\lfloor \frac{2}{3} n \rfloor-2}}{nW (2M)^{2n-1} {2n-1 \choose \lfloor \frac{2}{3} n \rfloor}}  .$
%   \[   \alphabw \leq 
% 1-\frac{2^{\lfloor \frac{2}{3} n \rfloor-2}}{nW (2M)^{2n-1} {2n-1 \choose \lfloor \frac{2}{3} n \rfloor}} \; . \]
If $\Gamma$ is unichain, the $d$-sensitive threshold $\alphad$ satisfies
$\displaystyle 
\alphad \leq 
1-\frac{2^{\min \{ \sens+2 , \lfloor \frac{2}{3} n - 1 \rfloor \} - 1}}{nW (2M)^{2n-1} {2n-1 \choose \min \{ \sens+2 , \lfloor \frac{2}{3} n - 1 \rfloor \} + 1}}  .$
% \[
% \alphad \leq 
% 1-\frac{2^{\min \{ \sens+2 , \lfloor \frac{2}{3} n - 1 \rfloor \} - 1}}{nW (2M)^{2n-1} {2n-1 \choose \min \{ \sens+2 , \lfloor \frac{2}{3} n - 1 \rfloor \} + 1}} \; .
% \]
\end{theorem}
Our next set of results are bounds on $\alphabw$ based on Mahler measures. 
 Since these bounds are more difficult to read than the ones in the previous theorem, we only give the resulting $O(\cdot)$ expressions here, and we provide the exact values in Sections~\ref{SectionDeterm} and~\ref{SectionStochastic}. We provide bounds on $-\log(1-\alphabw)$, i.e., on the value of $L>0$ such that $\displaystyle \alphabw \leq 1 - e^{-L}$.
\begin{theorem}[Based on Mahler measures]\label{ThMahler}
If the game $\Gamma$ satisfies \Cref{assumption}, then:

\noindent
{\em Deterministic case.} The Blackwell threshold $\alphabw$ satisfies
 \[- \log \left(1-\alphabw\right) \leq O\left(\max \left\{ \sqrt{n\log(n)\log(\sqrt{n}W)},\log(\sqrt{n}W) \right\} \right) \; . \]
% The next is the bound without the max
%\[\alphabw \leq 1 - e^{-%\left(O\left(\sqrt{n\log(n)\log(\sqrt{n}W)}\right)\right)}.\]
{\em Stochastic case.} The Blackwell threshold $\alphabw$ satisfies
 \[- \log \left(1-\alphabw\right) \leq O\Big(  \max \big\{ \log(W) + n(1+\log(M)),\sqrt{n\log(n)\left(\log(W) + n(1+\log(M))\right)} \big\} \Big) . \]
% The next pis the bound without the max
%\[\alphabw \leq 1 - e^{-\left( O \left( \sqrt{n\log(n)\left(\log(W) + n(1+\log(M))\right)} \right) \right)}.\]
\end{theorem}
We now present our last set of results. \cite{Veinott1969} shows that, for Markov decision processes (MDPs, i.e., the one-player case), $(n-2)$-sensitive optimal strategies are also Blackwell optimal. Our next theorem shows that a much smaller value of $d$ may suffice.
  \begin{theorem}\label{th:bound on k - intro}
  % \todo{jgc - addressed this + added exact th. in main body. SG: we could simplify this bound to $O(\sqrt{n\log W})$ since the constant $a$ is not explicit.}
%Assume that $\Gamma$ is a deterministic perfect-information stochastic game with $n \in \N$ states and integer rewards with absolute values bounded by $W \in \N$.
Assume the game $\Gamma$ satisfies \Cref{assumption} and is deterministic.
Then every $\bar{d}^{\de}(n,W)$-sensitive optimal strategy is Blackwell optimal, where $\bar{d}^{\de}(n,W) =O\left(\sqrt{n\log(W)}\right)$.
%%%%% jgc --- still to address for final submission
% \todo{RK: I modified the value of $\bar{d}^{\de}(n,W)$ subtracting $2$ from the previous version. Please check if it is correct (you can find a proof of this result in \Cref{SectionDeterm}). }
\end{theorem}
  When $\log(W) = o(n)$, $\bar{d}^{\de}(n,W)$ may be much smaller than $n-2$. Combining \Cref{th:bound on k - intro} with the bounds on $\alphad$ of \Cref{ThLagrange}, we arrive at the following bound on $\alphabw$.
 \begin{theorem}[Based on multiplicity]\label{th:bound on d}
 %Assume that $\Gamma$ is a deterministic perfect-information stochastic game with $n \in \N$ states and integer rewards with absolute values bounded by $W \in \N$. 
 Assume that $\Gamma$ satisfies \Cref{assumption} and is deterministic. 
 Then the Blackwell threshold $\alphabw$ satisfies $\displaystyle \alphabw \leq 1-\frac{1}{24 W{2n\choose {\min \{ O\left(\sqrt{n\log(W)}\right), n\}}}} \; .$
  % \begin{equation}\label{BEK bound}
  % \alphabw \leq 1-\frac{1}{24 W{2n\choose {\min \{ O\left(\sqrt{n\log(W)}\right), n\}}}} \; .
  % \end{equation}
  \end{theorem}  
In the case of general (non-deterministic) perfect-information SGs, the multiplicity method of \Cref{th:bound on d} does not lead to a useful bound. We further elaborate on this in \Cref{rmk:bound on d - stochastic} of \Cref{SectionStochastic}.
%Since we obtain our bounds from three different techniques, it may be difficult to compare our results at first sight.
To simplify the comparison between our bounds, and with existing work, we summarize all our results in Tables~\ref{tab:bound deterministic SGs} and~\ref{tab:bound general SGs}. In these tables, we reformulate our bounds on $\alphabw$ as bounds on $-\log(1-\alphabw)$, i.e., on the value of $L>0$ such that $\displaystyle \alphabw \leq 1 - e^{-L}$, and similarly for $\alphad$.

Note that Assumption~\ref{assumption} is necessary to obtain a meaningful bound on $\alphabw$ and $\alphad$ (Proposition 4.3 in \cite{grand2023reducing}).
We emphasize that the bounds derived in Theorems~\ref{ThLagrange},  \ref{ThMahler}, and~\ref{th:bound on d} are complementary, i.e., there are different regimes of $\log(W)$, $\log(M)$ and $n$ where one bound is better than the others. As an example, consider deterministic games ($M=1$). When $W=O(1)$, \Cref{ThLagrange} yields $-\log(1-\alphabw) \leq O(n)$, while Theorems~\ref{ThMahler} and~\ref{th:bound on d} both lead to the stronger bound $-\log(1-\alphabw) \leq O\left(\log(n)\sqrt{n}\right)$. In contrast, in the regime where $W = \exp( \Theta(n))$, Theorems~\ref{ThLagrange} and~\ref{th:bound on d} lead to the bound $-\log(1-\alphabw) \leq O(n)$, which is stronger than \Cref{ThMahler} yielding $-\log(1-\alphabw) \leq O\left(n \sqrt{\log(n)}\right)$. 
We provide more discussion on this in Sections~\ref{SectionDeterm} and~\ref{SectionStochastic}.
\todo{this paragraph needs to be rewritten to emphasize better the main implications of our results}Compared to previous work, we note that our bound on $-\log(1-\alphabw)$ from \Cref{ThLagrange} improves upon \cite{andersson2009complexity} and \cite{grand2023reducing} by a factor $\Omega(n)$, while our bound from Theorem \ref{th:bound on d} recovers for SGs the results obtained in \cite{mukherjee2025howard} for MDPs. \Cref{ThMahler} provides a new approach to bounding $\alphabw$ based on Mahler measures. We are also the first to provide bounds for $\alphad$, beyond the case $d=-1$ for deterministic games analyzed in~\cite{zwick}.
However, we note that our results on $\alphabw$ do not lead to weakly- or even quasi-polynomial time algorithms for computing Blackwell optimal strategies: our bounds on $\alphabw$ are still very close to $1$, and the algorithms with the best dependence on $\alpha$ for solving discounted SGs require $\tilde{O}((1-\alpha)^{-1})$ arithmetic operations~\cite{hansen2013strategy}. We refer to the next section for a detailed discussion and comparison with the existing bounds on $\alphabw$ and $\alphad$.

% \todo{SG: deleted lines in the table/ recommanded by nick higham handbook of writing, and this helps to fix overfull}
%The next is the table witout the max
%\begin{table}[htp]
%\caption{Bounds for {\em deterministic} SGs     satisfying \Cref{assumption}. For the sake of readability, in all these bounds we omit constant terms and the $O(\cdot)$ notation.}
%    \label{tab:bound deterministic SGs}
%    \centering
%    \begin{tabular}{@{}c@{}c|@{}c@{}@{}c@{}}
%\toprule
%      &Bound on $-\log(1-\alphabw)$ & Bound on $-\log(1-\alphad)$ & Remarks \\
%        \midrule
%     \cite{zwick} & $\times$ & $\log(W) +  3 \cdot \log(n)$ &  For $d=-1$\\
%    \cite{andersson2009complexity} & $n\log\left(n\right) + n^2 \log\left(W\right)$ & $\times$ &  $\cdot$\\
%        \cite{grand2023reducing} & $n^2 + n \log(W)$ & $\times$ & For MDPs \\
%        \cite{mukherjee2025howard} & $\sqrt{n \log(W)} \log\left(\frac{n}{\log(W)}\right)+ \log(W)$ & $\times$ & For MDPs \\
%        Th.~\ref{ThLagrange} &$n + \log(\frac{W}{\sqrt{n}})$ & $\log(W) + (d+4)\left(1+\log\left(\frac{2n}{d+4}\right)\right)$ & Lagrange \\
%        Th.~\ref{ThMahler} & $\sqrt{n\log(n)\log(\sqrt{n}W)}$ & $\times$ & Mahler  \\ 
%        Th.~\ref{th:bound on d} & $\sqrt{n \log(W)} \log\left(\frac{n}{\log(W)}\right)+ \log(W)$ & $\times$ &\ Multiplicity \\
%        \bottomrule
%    \end{tabular}
%\end{table}
%\todo{RK: I believe that it is recommended not to use vertical lines in tables}
\begin{table}[htp]
\caption{Bounds for {\em deterministic} SGs %with $n$ states in which $W \in \N$ is a bound on the absolute values of the integer rewards. 
    satisfying \Cref{assumption}. For the sake of readability, in all these bounds we omit constant terms and the $O(\cdot)$ notation.}
    \label{tab:bound deterministic SGs}
    \centering
    \begin{tabular}{@{}c@{}c@{}c@{}@{}c@{}}
%      References
\toprule
      &Bound on $-\log(1-\alphabw)$ & Bound on $-\log(1-\alphad)$ & Remarks \\
        \midrule
     \cite{zwick} & $\times$ & $\log(W) +  3 \cdot \log(n)$ &  For $d=-1$\\
    \cite{andersson2009complexity} & $n\log\left(n\right) + n^2 \log\left(W\right)$ & $\times$ &  $\cdot$\\
        \cite{grand2023reducing} & $n^2 + n \log(W)$ & $\times$ & For MDPs \\
        \cite{mukherjee2025howard} & $\sqrt{n \log(W)} \log\left(\frac{n}{\log(W)}\right)+ \log(W)$ & $\times$ & For MDPs \\
        Th.~\ref{ThLagrange} &$n + \log(\frac{W}{\sqrt{n}})$ & $\log(W) + (d+4)\left(1+\log\left(\frac{2n}{d+4}\right)\right)$ & Lagrange \\
        Th.~\ref{ThMahler} & $\max\{\sqrt{n\log(n)\log(\sqrt{n}W)}, \log(\sqrt{n}W) \}$ & $\times$ & Mahler  \\ 
        Th.~\ref{th:bound on d} & $\sqrt{n \log(W)} \log\left(\frac{n}{\log(W)}\right)+ \log(W)$ & $\times$ &\ Multiplicity \\
        \bottomrule
    \end{tabular}
\end{table}
%The next is the table witout the max
%\begin{table}[htp]
%    \caption{Bounds for SGs  satisfying \Cref{assumption}. For the sake of readability, in all these bounds we hide the constant term and the $O(\cdot)$ notation. $^\dagger$: This bound holds only in the unichain case.}
%    \label{tab:bound general SGs}
%    \centering
%    \begin{tabular}{@{}c@{}c|@{}c@{}@{}c@{}}
%\toprule
%    & Bound on $-\log(1-\alphabw)$ & Bound on $-\log(1-\alphad)$ & Remarks \\
%        \midrule
%    \cite{andersson2009complexity} & $n\log\left(n\right) + n^2 \log\left(\max\{W,M\}\right)$ & $\times$ &  $\cdot $\\
%        \cite{grand2023reducing} & $n^2(1+\log(M))+ n \log(W) $ & $\times$ & For MDPs \\
%        Th.~\ref{ThLagrange} & $n(1+\log(M))+ \log(W) $ & $\log(nW (2M)^{2n-1}) + (d+4)\log\left(\frac{2n}{d+4}\right)${$^{\normalsize\dagger}$}  & \ Lagrange \\
%        Th.~\ref{ThMahler} & $\sqrt{n\log(n)(n(1+\log(M)) + \log(W))}$ & $\times$ & Mahler \\
%        \bottomrule
%    \end{tabular}
%\end{table}
\begin{table}[htp]
    \caption{Bounds for SGs 
    satisfying \Cref{assumption}. For the sake of readability, in all these bounds we hide the constant term and the $O(\cdot)$ notation. $^\dagger$: This bound holds only in the unichain case.}
    \label{tab:bound general SGs}
    \centering
    \begin{tabular}{@{}c@{}c@{}c@{}@{}c@{}}
\toprule
    & Bound on $-\log(1-\alphabw)$ & Bound on $-\log(1-\alphad)$ & Remarks \\
        \midrule
    \cite{andersson2009complexity} & $n\log\left(n\right) + n^2 \log\left(\max\{W,M\}\right)$ & $\times$ &  $\cdot $\\
        \cite{grand2023reducing} & $n^2(1+\log(M))+ n \log(W) $ & $\times$ & For MDPs \\
        Th.~\ref{ThLagrange} & $n(1+\log(M))+ \log(W) $ & $\log(nW (2M)^{2n-1}) + (d+4)\log\left(\frac{2n}{d+4}\right)${$^{\normalsize\dagger}$}  & \ Lagrange \\
        Th.~\ref{ThMahler} & $\max\{\sqrt{n\log(n)F},F\}$ & $\times$ & Mahler \\ 
        & $ F :=  n(1+\log(M)) + \log(W)$ & & \\
        \bottomrule
    \end{tabular}
\end{table}
\paragraph{Related work.}
%%%%% jgc --- still to address for final submission
% \todo[inline]{SG: we should include applications of the approach. we find the same threshold as Zwick and paterson for $-1$-otpimality, but we dont need their rounding argument: it sufficies to take optimal policies for the discounted game. For generic input, convergence of policy iteration for mean payoff games (Cochet, Gunawardena, SG) simulates the Blackwell optimal discounted policy iteration. So we get a pseudo polynomial bound for the  policy iteration algorithm in CGG for a generic input. }
% \todo[inline]{SG: we distinguish the blackwell threshold for the $1$-threshold. the $-1$-threshold is the one which governs complexity bounds. we can recover as a corollary of ye/zwick,hansen,miltersen/marianne,sg a weakly polynomial bound from policy iteration. Write it.}
To the best of our knowledge, only a small number of papers have obtained bounds on the $\sens$-sensitive threshold $\alphad$ (only for $d=-1$) or on the Blackwell threshold $\alphabw$. 

\noindent 
{\em $\sens$-sensitive threshold.} The existence of $\alpha_{-1}$ is proved in the seminal paper~\cite{liggettlippman} and in Theorem~5.4.4 of~\cite{puri1995theory}. It is used, for instance, in \cite{andersson2009complexity} to show that these games polynomially reduce to solving discounted SGs.
\cite{zwick} shows the tight bound $\alpha_{-1} \geq 1 - \frac{1}{8 W n^3}$ for deterministic perfect-information mean-payoff SGs. \Cref{ThLagrange} recovers this bound for $\sens=-1$, but no general bound was known for $\alphad$ (before our work). \cite{akian2013policy} shows that if there is a renewal state, i.e., a state for which the time of first return from any other state is bounded by $N \in \N$, then we can take $\alpha_{-1}= 1 - 1/N$. A characterization of $\alpha_{-1}$ for the one-player case (i.e., for MDPs) is given in \cite{boone2023discounted} but the obtained bound relies on the minimum gain difference between two strategies, instead of only the game parameters $n,M,W$ as in our results. Finally, we note that a line of work studies a related but different question for MDPs and reinforcement learning, namely, how close to $1$ should $\alpha$ be for a discount optimal strategy to be $\epsilon$-optimal for the mean-payoff criterion, e.g.~\cite{wang2022near} showing that $\alpha \geq 1 - \epsilon/H$ suffices for weakly-communicating instances with diameter $H \in \N$, or~\cite{jin2021towards} using assumptions on the mixing times of the Markov chains induced by deterministic strategies. In contrast to these works, we focus on optimality (instead of $\epsilon$-optimality) and most of our results do not require an assumption on the chain structure (except for $\alphad$ in the stochastic case). 

\noindent 
{\em Blackwell threshold.}
The existence of $\alphabw$ comes from an argument developed initially by Blackwell~\cite{blackwell1962discrete}. The works closest to ours are \cite{andersson2009complexity,grand2023reducing,mukherjee2025howard}. \cite{andersson2009complexity} shows a bound on $\alpha_{-1}$ for SGs but their argument extends directly to $\alphabw$.
\cite{grand2023reducing} gives a bound on $\alphabw$ for MDPs.
\cite{mukherjee2025howard} shows an analog of \Cref{th:bound on k - intro} for {\em deterministic} MDPs using Lagrange's bound. We summarize these bounds in Tables~\ref{tab:bound deterministic SGs} and~\ref{tab:bound general SGs}. We also note that the Blackwell threshold has found recent applications in the study of robust MDPs~\cite{grand2023beyond,wang2024robust} and in the smoothed analysis of the complexity of strategy iteration~\cite{loff2024smoothed}.

\noindent
{\em Main improvements compared to previous works.} Our work appears to be the first to provide a bound on the $\sens$-sensitive threshold $\alphad$ beyond the case $\sens=-1$ (corresponding to mean-payoff optimality) in the deterministic setting. Regarding $\alphabw$, \cite{mukherjee2025howard} only focuses on deterministic MDPs using the multiplicity approach, i.e., on the one-player deterministic case, whereas we focus on two-player SGs. This line of analysis based on multiplicity (\Cref{th:bound on k - intro}) does not extend to the non-deterministic case (see \Cref{rmk:bound on d - stochastic}), where the corresponding bound on the value of $\sens$ such that $\alphad = \alphabw$ is much larger than $n-2$, the bound proved in the seminal paper~\cite{Veinott1969}.
Compared to \cite{grand2023reducing} and \cite{andersson2009complexity}, our bounds from Theorems~\ref{ThLagrange} and~\ref{ThMahler} compare favorably: for instance, the bounds on $-\log(1-\alphabw)$ in \cite{grand2023reducing} and \cite{andersson2009complexity} are worse by a factor of $\Omega(n)$ compared to our bounds based on Lagrange. This improvement comes from two main innovations in the analysis: better bounds on the magnitudes of the coefficients of the considered polynomials, and the use of much stronger root separation results - \cite{andersson2009complexity} uses a bound attributed to Cauchy (e.g. Equation (5) in Lecture IV of \cite{yap2000fundamental}) which is weaker than the one we use due to Lagrange, while \cite{grand2023reducing} uses a bound from \cite{rump1979polynomial}, which applies to the distance between any two roots of a polynomial, whereas we only need to separate a root from $1$ (and not from any other conjugates).
%% Finally, we note that the bound of $\alpha_d$ in Theorem~\ref{ThLagrange} implies that for a {\em fixed} value of $d$, deterministic mean-payoff games can be solved in pseudo-polynomial time, this result was only known when $d=-1$,
%% as a consequence of the analysis of value iteration by Zwick and Paterson~\cite{zwick}.
%%%%% jgc --- still to address for final submission
% \todoi{jgc: are we missing a short literature review on algorithms/complexity for (a) discounted optimality (b) mean-payoff (c) $\sens$-sensitive optimality (d) Blackwell optimality?}

{\bf Outline.} This paper is organized as follows. We introduce perfect-information SGs in \Cref{sec:stochastic games}. We derive bounds for the deterministic case in \Cref{SectionDeterm} and bounds for the stochastic case in \Cref{SectionStochastic}. The implications of our results are discussed in \Cref{sec:discussion}. All the proofs can be found in the appendix.
\section{Preliminaries on perfect-information stochastic games}\label{sec:stochastic games}
Lloyd Shapley introduced SGs in the seminal paper~\cite{shapley1953stochastic}. We focus on the case of two-player zero-sum SGs with finitely many states and actions, and discrete time. We denote the state space by $[n] \coloneqq \{1,\dots,n\}$ for some $n \in \N$. At every stage $k \in \N$, the game is in a state $i_k \in [n]$ observable by both players, called $\Min$ and $\Max$. An instantaneous reward of $r_{i}^{ab}$ is determined when $\Min$ chooses action $a$ and $\Max$ chooses action $b$ in state $i$, and then the game transitions to the next state $j \in [n]$ with probability $P^{ab}_{ij} \in [0,1]$. We focus on the case of {\em perfect-information} SGs, where we can partition the state space $[n]$ into the states controlled by $\Min$ and those controlled by $\Max$. 

A {\em strategy} of a player is a function that assigns to a history of the game (i.e., the sequence of previous states and actions) a decision (choice of an action) of this player. %A strategy is said to be {\em positional} when the chosen action depends only on the current state. 
A pair of strategies $(\sigma,\tau)$ of players $\Min$ and $\Max$ induces a probability measure on the set of sequences of states.
We define the {\em discounted value function} $\alpha \mapsto v^{\sigma,\tau}_{i}(\alpha)$ from $]0 , 1[$ to $\R$, associated with the pair of strategies $(\sigma,\tau)$ and the initial state $i$, as
$\displaystyle
v^{\sigma,\tau}_{i}(\alpha) := \E^{\sigma,\tau} \left[ \sum_{k=1}^{+\infty} \alpha^k r_{i_k}^{a_k b_k} \; | \; i_0 = i \right]
  $
  % \[
  % v^{\sigma,\tau}_{i}(\alpha) = \E^{\sigma,\tau} \left[ \sum_{k=1}^{+\infty} \alpha^k r_{i_k}^{a_k b_k} \; | \; i_0 = i \right]
  % \]
where $r_{i_1}^{a_1 b_1},r_{i_2}^{a_2 b_2},\ldots $ is the random sequence of instantaneous rewards induced by $(\sigma,\tau)$. 
% \todoi{should we start value function at time $k=0$ or $k=1$?}
\begin{definition}
      A pair of strategies $(\sigma^*,\tau^*)$ of players Min and Max is {\em discount optimal} for the discount factor $\alpha$ if for any state $i$ and for any strategies $\sigma$ and $\tau$ of players Min and Max, we have
\[
v^{\sigma,\tau\opt}_{i}(\alpha) \geq v^{\sigma\opt,\tau\opt}_{i}(\alpha) \geq v^{\sigma\opt,\tau}_{i}(\alpha)\; .
\]

A pair of strategies $(\sigma^*,\tau^*)$ is {\em Blackwell optimal} if there exists $\bar{\alpha} <1$ such that $(\sigma^*,\tau^*)$ is discount optimal for all discount factors larger than $\bar{\alpha}$.
\end{definition}
Shapley~\cite{shapley1953stochastic} shows the existence of stationary discount optimal strategies. If, in addition, the game is a perfect-information SG, then these stationary optimal strategies may be chosen deterministic. The existence of stationary deterministic Blackwell optimal strategies for perfect-information SGs is a consequence of the analysis in the seminal paper~\cite{liggettlippman} (see the proof of Theorem 1 in \cite{liggettlippman}).
%Given two pairs of stationary strategies $(\sigma,\tau)$ and $(\sigma',\tau')$, we write $(\sigma,\tau) \succeq_{\sens} (\sigma',\tau')$ if 
%\begin{equation}\label{Eq:Def-d-sensitive}
%\lim_{\alpha\to 1^-} \; (1-\alpha)^{-\sens}\left(v^{\sigma,\tau}_{i}(\alpha)-v^{\sigma',\tau'}_{i}(\alpha) \right) \geq 0 
%\end{equation}
%for any state $i$. We are now ready to define $\sens$-sensitive optimality.
\begin{definition}
  A pair of strategies $(\sigma^*,\tau^*)$ of players Min and Max is {\em $\sens$-sensitive optimal} if
%\[
%(\sigma,\tau^*) \succeq_{\sens} (\sigma^*,\tau^*) \succeq_{\sens} (\sigma^*,\tau) \enspace ,
%\]
\begin{equation}\label{Eq:Def-d-sensitive}
\lim_{\alpha\to 1^-} \; (1-\alpha)^{-\sens}\left(v^{\sigma,\tau^*}_{i}(\alpha)-v^{\sigma^*,\tau^*}_{i}(\alpha) \right) \geq 0 ,
\lim_{\alpha\to 1^-} \; (1-\alpha)^{-\sens}\left(v^{\sigma^*,\tau^*}_{i}(\alpha)-v^{\sigma^*,\tau}_{i}(\alpha) \right) \geq 0 
\end{equation}
for any state $i$ and any strategies $\sigma$ and $\tau$ of players Min and Max respectively.
\end{definition}
Our definition of $\sens$-sensitive optimality for two-player SGs recovers the definition of $\sens$-sensitive optimality for the one-player case~\cite{Veinott1969}. Note that Blackwell optimal strategies are $\sens$-sensitive optimal for $\sens =-1,0,\dots$. In fact, the same analysis as for MDPs (using the Laurent series expansion, e.g. Theorem~10.1.6 of \cite{puterman2014markov}) shows that $n-2$-sensitive optimal strategies are Blackwell optimal, and mean-payoff optimality corresponds to $d=-1$.
%In the rest of the paper, we will assume that there exists $W \in \N$ and $M \in \N$ such that for any states $i,j \in [n]$ and any pair of actions $(a,b)$, the rewards satisfy $|r_{i}^{ab}| \leq W$ and the transition probabilities $P^{ab}_{ij} \in [0,1]$ satisfy $M \cdot P^{ab}_{ij} \in \N$.

Before diving into our main results, we give an {\bf overview of our main proof techniques}.

{\em Bounds on $\alphabw$.}
Our bounds are based on studying the zeros of the functions $\alpha \mapsto v^{\sigma,\tau}_{i}(\alpha) - v^{\sigma',\tau'}_{i}(\alpha)$ for any pairs $(\sigma,\tau)$ and $(\sigma',\tau')$ of stationary strategies and any state $i$. This approach has been used for MDPs, e.g., see Theorem~2.16 of~\cite{feinberg2012handbook}.
Differences of discounted value functions %$\alpha \mapsto v^{\sigma,\tau}_{i}(\alpha) - v^{\sigma',\tau'}_{i}(\alpha)$
are rational functions in $\alpha \in ]0,1[$, i.e., they can be written as the ratio of two polynomials in $\alpha$. As such, each of them can only have finitely many zeros in $]0,1[$. Additionally, for perfect-information SGs, discount optimal strategies may be chosen stationary and {\em deterministic}~\cite{gillette1957stochastic}, hence we can consider only finitely many pairs of strategies. Therefore, there exists a discount factor $\bar{\alpha} < 1$ such that none of the functions $\alpha \mapsto v^{\sigma,\tau}_{i}(\alpha) - v^{\sigma',\tau'}_{i}(\alpha)$ has a zero in $]\bar{\alpha},1[$, and so {\em all} these functions have constant sign in $]\bar{\alpha},1[$. This guarantees that $\alphabw \leq \bar{\alpha}$, since strategies that are discount optimal for some $\alpha$ satisfying $\bar{\alpha} < \alpha < 1$ remain discount optimal for all larger discount factors (otherwise some function $\alpha \mapsto v^{\sigma,\tau}_{i}(\alpha) - v^{\sigma',\tau'}_{i}(\alpha)$ would change sign in $]\bar{\alpha},1[$). 
%As evident from the previous paragraph, $\alphabw$ is a zero of a function $\alpha \mapsto v^{\sigma,\tau}_{i}(\alpha) - v^{\sigma',\tau'}_{i}(\alpha)$ (for some $(\sigma,\tau),(\sigma',\tau')$)\todo{should we add a proof of this in the appendix}, and to bound $\alphabw$ the main difficulty is to determine how close to $1$ can a zero of such a function be. 
Thus, a way to bound $\alphabw$ is to determine how close to $1$ can a zero of the functions $\alpha \mapsto v^{\sigma,\tau}_{i}(\alpha) - v^{\sigma',\tau'}_{i}(\alpha)$ be. 
Given \Cref{assumption}, we can bound the degree of the numerator of $\alpha \mapsto v^{\sigma,\tau}_{i}(\alpha) - v^{\sigma',\tau'}_{i}(\alpha)$ and the size of its coefficients. We then use {\em root separation} results to separate the root of a polynomial from a given scalar.
Compared to previous work, our first improvement lies in a better analysis of the coefficients of the considered polynomials. As our second and main improvement, we use stronger separation bounds than in previous work, a bound due to Lagrange (\Cref{th:Lagrange}) and a bound based on Mahler measures (\Cref{th:mahler}). 
\cite{andersson2009complexity} uses a weaker bound due to Cauchy, and~\cite{oliu2021new} also uses the Cauchy bound for a different purpose (bounding the variations in value functions when $\alpha \rightarrow 1$).

\noindent
{\em Deterministic vs. non-deterministic games.} For {\em deterministic} SGs, each transition probability belongs to $\{0,1\}$. In this case, the discounted value function associated with a pair of stationary strategies and a state can be represented by the concatenation of a path and an elementary circuit in the graph of the game, see Section~\ref{SectionDeterm}. This can be used to analyze the degree and the magnitude of the coefficients of the numerator of $\alpha \mapsto v^{\sigma,\tau}_{i}(\alpha) - v^{\sigma',\tau'}_{i}(\alpha)$. When the game is not deterministic, we rely on the closed-form expression of the discounted value functions using cofactor matrices, see Section~\ref{SectionStochastic}.

\noindent
{\em Bounds on the $\sens$-sensitive threshold $\alphad$.} Our bounds on $\alphad$ rely on the same proof techniques as for $\alphabw$, %except that the limits~\eqref{Eq:Def-d-sensitive} provide more information on the coefficients of the considered polynomials. 
except that the $d$-sensitive optimality allows us to deduce more information on the coefficients of the considered polynomials (more precisely, that some of them are zero). \todo{need more explanations here. RK: I slightly modified the sentence}
%Again, the deterministic case is easier because we can exploit the structure of the discounted value functions (since every run of the game is the concatenation of a path and then an elementary circuit repeated infinitely), whereas for the case of non-deterministic SGs we require a unichain assumption.
Again, the deterministic case is easier because we can exploit the graph representation  of the discounted value functions mentioned above, whereas for the case of non-deterministic SGs we require a unichain assumption.

\section{Results for perfect-information deterministic stochastic games}\label{SectionDeterm}
We start by analyzing the structure of the discounted value functions in deterministic perfect information SGs. A deterministic game can be represented by a weighted digraph: its set of nodes is the set of states, %controlled by player $\Min$ (if we assume that $\Min$ plays first), 
there is an edge from state $i$ to state $j$ if some actions $a$ and $b$ of the players realize this transition,
% $P^{ab}_{ij}=1$ for some actions $a$ and $b$ of players $\Min$ and $\Max$ respectively,
in which case this edge has weight $r^{ab}_{ij}$.  
Given an initial state $i$, a pair of stationary strategies $(\sigma,\tau)$ determines a run of the game of the form $``\pi\gamma"$, in which $\pi$ is a path and $\gamma$ is a circuit. %is an elementary path, and $\gamma$ is an elementary circuit (i.e., a sequence of states where only the first and last elements are equal). The {\em length} of a circuit is equal to the number of states in the circuit, minus $1$. 
%Denoting by $p$ the length of $\pi$ (i.e., the number of edges it contains) and $q$ the length of $\gamma$, we can choose $\pi$ and $\gamma$ such that $p+q\leq n$ and $p\leq n-1$. 
Given a path $\pi=(i_0,\dots,i_{k})$, we set $\<r,\pi> := r_{i_0i_1} + \alpha r_{i_1i_2} + \dots + \alpha^{k-1}r_{i_{k-1}i_k}$ (for simplicity, we denote the instantaneous rewards by $r_{ij}$ instead of $r^{ab}_{ij}$). 
% and for $q \in \N$, we set $\q[q]:=1+\alpha+ \dots + \alpha^{q-1}$. 
%Then, the two discounted value functions $v_i^{\sigma,\tau}(\alpha)$ and $v_i^{\sigma',\tau'}(\alpha)$ are given by $v_i^{\sigma,\tau} (\alpha) = \<r,\pi> + \frac{\alpha^p}{1-\alpha^q} \<r,\gamma>$ and $v_i^{\sigma',\tau'} (\alpha) = \<r,\pi'> + \frac{\alpha^{p'}}{1-\alpha^{q'}} \<r,\gamma'>$ respectively. Let us consider the polynomial $\Delta(\alpha)$ corresponding to the denominator in $v_i^{\sigma, \tau} (\alpha) - v_i^{\sigma',\tau'} (\alpha)$:\todo{jgc: in next equation, removed the notation [q] and just keep it in the proof in the appendix}
With this notation, $v_i^{\sigma,\tau}(\alpha)=\<r,\pi> + \frac{\alpha^p}{1-\alpha^q} \<r,\gamma>$, where $p$ is the length of $\pi$ (i.e., the number of edges it contains) and $q$ is the length of $\gamma$. 
To study the zeros of $\alpha \mapsto v^{\sigma,\tau}_{i}(\alpha) - v^{\sigma',\tau'}_{i}(\alpha)$, we analyze the roots of the polynomial 
\begin{equation}\label{eq:def_Delta}
\Delta(\alpha) := (1-\alpha^q)(1-\alpha^{q'}) (v_i^{\sigma, \tau} (\alpha) - v_i^{\sigma',\tau'} (\alpha) ) 
\end{equation}
(here $q'$ is the length of the circuit associated with $(\sigma',\tau')$). 
%In this paper, to bound $\alphabw$ we consider pairs of strategies $(\sigma,\tau)$ and $(\sigma',\tau')$ such that $\alphabw$ is a zero of $\Delta: \Delta(\alphabw)=0$, and to bound $\alphad$, we choose $\Delta$ for which $\Delta(\alphad)=0$.\todo{\jgc{added this sentence to address Ricardo's point right after th 3.7}} 
The next lemma bounds the degree and coefficients of $\Delta$.
\begin{lemma}\label{lem:degree coeff Delta - deterministic}
We can write $\Delta(\alpha) =
\sum_{k=0}^K a_k \alpha^k$,
where $|a_k|\leq 12 W$ and $K \leq 2n-1$.
\end{lemma}
% A few lines of algebra show that we can write $\Delta(\alpha)=\sum_{k=0}^K a_k \alpha^k$, where $|a_k|\leq 12 W$ and $K \leq 2n-1$.
As described in the previous section, our bounds on $\alphabw$ and $\alphad$ rely on separating the roots of the polynomial $\Delta$ from $1$. %and on obtaining good bounds on the size of the coefficients of $\Delta$. 
We proceed to do so in the next section.

\subsection{Separation based on the Lagrange bound}\label{sec:determ Lagrange}

In this section, we rely on the {\em Lagrange bound} (e.g.~Lemma 5 in Lecture IV of~\cite{yap2000fundamental}).
\begin{theorem}[Lagrange bound]\label{th:Lagrange}
Let $P = \sum_{k=j}^{d} c_k x^k$ with $c_j \neq 0$. Then, any non-zero root $z$ of $P$ satisfies 
 $
 | z | \geq \frac{1}{2}\min_{i \in \{ j+1, \ldots , d\}, c_i \neq 0} \left(|c_{j}|/|c_i|\right)^{\frac{1}{i-j}} $.
\end{theorem}
% \todo{next prop could go in appendix. RK: done}
%This bound can be used to separate the non-zero roots of a polynomial $P$ from $0$. Since we want to separate the roots of $\Delta$ from $1$, we will apply this bound to the polynomial $\epsilon \mapsto \Delta(1-\epsilon)$. %Note that
%\begin{align}\label{DeltaEpsilon}
%  \Delta(1-\epsilon)  
  % =   \sum_{k=0}^K a_k (1-\epsilon)^k
   % = \sum_{k=0}^K a_k (\sum_{i=0}^k {k\choose i}(-1)^i\epsilon^i)
%   = \sum_{i=0}^K (-1)^i\epsilon^i b_i \enspace 
   % \text{ where } b_i &= \sum_{k=i}^K  a_k {k\choose i}\enspace .
%\end{align}
%where $b_i = \sum_{k=i}^K  a_k {k\choose i}$. 
%\begin{proposition}\label{prop_threshold_det}
%Suppose that $b_0 = \ldots = b_{j-1} = 0$ and $b_j \neq 0$ for some $j \geq 1$ in~\eqref{DeltaEpsilon}. Then, the polynomial $\epsilon \mapsto \Delta(1-\epsilon)$ has no zeros in the interval $]0 , \frac{1}{24 W{K+1\choose j+2}}[$.
%\end{proposition} 
%Using this proposition we prove the first part of \Cref{ThLagrange}, namely, the bound on the $d$-sensitive threshold $\alphad$ for deterministic perfect-information SGs.
The Lagrange bound separates the non-zero roots of a polynomial $P$ from $0$. To separate the roots of $\Delta$ from $1$, we apply this bound to the polynomial $\epsilon \mapsto \Delta(1-\epsilon)$. Using this approach, we prove the first part of \Cref{ThLagrange}, namely, the bound on $\alphad$ for deterministic perfect-information SGs.
\begin{theorem}\label{th:bound $\sens$-sensitive discount factor - deterministic}
%Assume that $\Gamma$ is a deterministic perfect-information stochastic game with $n \in \N$ states and integer rewards with absolute values bounded by $W \in \N$. 
Assume the game $\Gamma$ satisfies \Cref{assumption} and is deterministic ($M = 1$). 
Then, the $\sens$-sensitive threshold $\alphad$ satisfies $\displaystyle \alphad \leq 1-\frac{1}{24 W{2n\choose {\min \{ \sens+4 , n\}}}} \; .$
% \[\alphad \leq 1-\frac{1}{24 W{2n\choose {\min \{ \sens+4 , n\}}}} \; .\]
\end{theorem}
The binomial coefficient in our bound on $\alphad$ appears because of the change of variable $\epsilon = 1-\alpha$ in the polynomial $\Delta$, necessary to apply the Lagrange bound to the polynomial $\epsilon \mapsto \Delta(1-\epsilon)$.
%% We recover the result for $\sens=-1$ of~\cite{zwick} (known to be tight), and we obtain a bound on $\alphabw$ using the fact that a pair of strategies that is $d$-sensitive optimal for all $d$ is also Blackwell optimal
Together with the result of~\cite{HansenICS2011}, our bound of $\alpha_d$
implies that for a fixed value of $d$, we can compute
$d$-sensitive optimal policies of a deterministic game
in pseudo-polynomial time, extending
a theorem of~\cite{zwick} ($d=-1$ case, then the bound is optimal).

\begin{corollary}\label{coroZP96}
%Assume that $\Gamma$ is a deterministic perfect-information stochastic game with $n \in \N$ states and integer rewards with absolute values bounded by $W \in \N$. 
If the game $\Gamma$ satisfies \Cref{assumption} and is deterministic ($M = 1$), we have  $\displaystyle \alpha_{-1} \leq 1-\frac{1}{O(Wn^3)}$ and $\displaystyle\alphabw \leq 1-\frac{1}{24 W{2n\choose n}}$.
% \[
%   \alpha_{-1} \leq 1-\frac{1}{O(Wn^3)}\; \makebox{ and }\;  \alphabw \leq 1-\frac{1}{24 W{2n\choose n}}.
%   \]
  \end{corollary}

{\bf Multiplicity approach.} In this approach, we combine our bound on $\alphad$ from \Cref{th:bound $\sens$-sensitive discount factor - deterministic} with a new result on the smallest integer $\sens$ such that $\sens$-sensitive optimal strategies are also Blackwell optimal. Our next theorem parametrizes the value of such $d$ by $n$ and $W$.
  \begin{theorem}\label{th:bound on k}
Assume the game $\Gamma$ satisfies \Cref{assumption} and is deterministic.
Then, there exists a constant $a>0$ such that $\bar{d}^{\de}(n,W)$-sensitive optimal strategies are Blackwell optimal, where $\bar{d}^{\de}(n,W) \coloneqq a \sqrt{(2n-1)(1+\log(12 W))}-2$.
\end{theorem}
Note that $\bar{d}^{\de}(n,W) = O\left(\sqrt{n(1+\log(W)}\right)$, which may be much smaller than $n-2$ (proved by~\cite{Veinott1969} for MDPs) in some regimes where $\log(W) = o(n)$. To show that $d=\bar{d}^{\de}(n,W)$ suffices for Blackwell optimality, we show that if a pair of strategies is $\bar{d}^{\de}(n,W)$-sensitive optimal, it is also $d$-sensitive optimal for all $d \geq \bar{d}^{\de}(n,W)$, so that it is $d$-sensitive optimal for $d=-1,0,\dots$, therefore it is Blackwell optimal. A key ingredient in our proof is a bound on the multiplicity of $1$ as a root of a polynomial as a function of its degree and the size of its coefficients~\cite{BEK1999}, used to bound the multiplicity of $1$ as a root of $\Delta(\alpha)$. \todo{should we explain why the multiplicity at $1$ matters?}
  \Cref{th:bound on d} follows directly from the bound on $\alphad$ of \Cref{th:bound $\sens$-sensitive discount factor - deterministic} by choosing $d= \bar{d}^{\de}(n,W)$. We refer to Table \ref{tab:bound deterministic SGs} for comparisons between the bounds obtained in this section, which improve by $\Omega(n)$ (for $-\log(1-\alphabw)$) compared to previous works.
\subsection{Separation based on Mahler measures}\label{sec:determ mahler measures}

We now present a bound on $\alphabw$ based on the {\em Mahler measure} of algebraic numbers~\cite{lehmer1933factorization,mahler1962some}.
\begin{definition}
%Let $z \in \R$ be an algebraic number of degree $d \in \N$, with complex conjugates $z_1 = z,z_2,\ldots ,z_d$, and minimal polynomial $P=a  \prod_{i=1}^{d}(x-z_{i})$. The Mahler measure $M(z) \in \R$ of $z$ is defined as $M(z) = a \prod_{i=1}^{d} \max \{1,|z_j|\}$. 
%
%More generally, the Mahler measure of a polynomial $P$ with complex coefficients is given by $M(P) = a  \prod_{i=1}^{d} \max \{1,|z_j|\}$ if $P$ factorizes over the complex numbers as $P=a \prod_{i=1}^{d}(x-z_{i})$.
The Mahler measure $M(P) \in \R$ of a polynomial $P$ is given by $M(P) := a  \prod_{i=1}^{d} \max \{1,|z_j|\}$ if $P$ factorizes over the complex numbers as $P=a \prod_{i=1}^{d}(x-z_{i})$. The Mahler measure $M(z) \in \R$ of an algebraic number $z$ is the Mahler measure of its minimal polynomial.
\end{definition}
Mahler measures have applications in areas like polynomial factorization, Diophantine approximation, and knot theory, see~\cite{smyth2008mahler}.
Mignotte and Waldschmidt~\cite{mignotte1994algebraic}
and Dubickas~\cite{Dubickas1995} use them to separate algebraic numbers and $1$. %The following theorem is an adaptation of the results in~\cite{Dubickas1995}.
%\begin{theorem} \label{th:mahler}
%There exists a universal constant $B>0$, such that any algebraic number $z$ of degree $d$ which is not a root of unity satisfies
% $|z -1 | > e^{-B \cdot \sqrt{d \log d  \log M(z)}}.$
%\begin{equation}\label{dubickas-0}
%|z -1 | > e^{-B  \sqrt{d \log d  \log M(z)}}.
%\end{equation}
%\end{theorem}
 The following is Theorem~1 of~\cite{Dubickas1995}.  
 \begin{theorem}%[{\cite[Theorem~1]{Dubickas1995}]
 \label{th:mahler}
 Let $\epsilon>0$. 
 There exists a constant $D_\epsilon \in \N$ such that any algebraic number $z$ of degree $d > D_\epsilon$ which is not a root of unity %and minimal polynomial $P =\sum^{d}_{k=0} c_k x^k$ 
 satisfies:
  \begin{equation}\label{dubickas}
 |z -1 | > e^{-(\pi/4 + \epsilon) \sqrt{d \log d  \log M(z)}}.
 \end{equation}
 \end{theorem}
To apply \Cref{th:mahler}, we need an estimate of $M(z)$, where $z$ is a root of $\Delta$. % Since $\Delta(\alphabw)=0$,\todo{RK: are we sure that $\alphabw$ is a root of $\Delta$? \jgc{$\alphabw$ should be the zero of ``a" function $\Delta$. I added a sentence in that regards right before Lemma 3.1.}} the minimal polynomial of $\alphabw$ divides $\Delta$, and $M(\alphabw) \leq M(\Delta)$. 
Since $\Delta(z)=0$, the minimal polynomial of $z$ divides $\Delta$, and so $M(z) \leq M(\Delta)$. 
A classical result of Landau~\cite{Landau1905} states that for any complex polynomial $P =\sum^{d}_{k=0} c_k x^k$ we have $M(P) \leq \sqrt{\sum^{d}_{k=0} |c_k|^2}$.
Applying this to the polynomial $\Delta$ (whose coefficients and degree are bounded in \Cref{lem:degree coeff Delta - deterministic}), we arrive at the following upper bound for $\alphabw$. 
% statement without the max
%\begin{theorem}\label{th-dubickas} 
%If the game $\Gamma$ satisfies \Cref{assumption} and is deterministic, there exists a universal constant $C>0$ such that the Blackwell threshold $\alphabw$ satisfies $\alphabw \leq \alphama^{\sf det}$, where 
% \[\alphama^{\sf det} = 1 - \frac{1}{\exp\left(B \cdot \sqrt{(2n-1) \log(2n-1) \log ( 12 \sqrt{2} \sqrt{n} W)}\right) }.\]
%\[\alphama^{\sf det} \coloneq 1 - e^{-C  \sqrt{n\log(n) \log(\sqrt{n} W)} } \; .\]
%\end{theorem}
\begin{theorem}\label{th-dubickas} 
If the game $\Gamma$ satisfies \Cref{assumption} and is deterministic, for each $\epsilon > 0$ there exists a constant $a_\epsilon>0$ such that the Blackwell threshold $\alphabw$ satisfies $\alphabw \leq \alphama^{\sf det}$, where 
% \[\alphama^{\sf det} = 1 - \frac{1}{\exp\left(B \cdot \sqrt{(2n-1) \log(2n-1) \log ( 12 \sqrt{2} \sqrt{n} W)}\right) }.\]
%\[\alphama^{\sf det} \coloneq 1 - e^{-C  \sqrt{n\log(n) \log(\sqrt{n} W)} } \; .\]
\[  - \log (1-\alphama^{\sf det}) = \max \Big\{
(\frac{\pi}{4} + \epsilon)\sqrt{(2n-1) \log(2n-1) \log ( 12 \sqrt{2} \sqrt{n} W)}
   ,  a_\epsilon +\log( 12 \sqrt{2} \sqrt{n}W) \Big\} \; .
\]
\end{theorem}
 The maximum in \Cref{th-dubickas} is necessary because for a fixed $\epsilon>0$, the bound in \Cref{th:mahler} only applies to algebraic numbers with a degree greater than $D_{\epsilon}$. Thus, if a root of $\Delta$ has degree greater than $D_{\epsilon}$, we can apply \Cref{th:mahler}, otherwise we apply the Lagrange bound.
%It is possible to obtain a bound with an explicit value for the constant $C$ (in the order of $\pi/4+\epsilon$ for any $\epsilon>0$), but this complicates the exposition of our results. We describe this in Appendix \ref{app:detailed Dubickas}.
\begin{remark}\label{rk-regimes}
  When $W$ is ``small'' ($W=n^{O(1)}$), we obtain from the definition of $\alphama^{\sf det}$ that 
  $-\log (1-\alphama^{\sf det}) = O( \sqrt{n} \log n)$. When $W$ is ``large'' ($W=\exp(\Omega(n \log^2 n))$), we have $ -\log (1-\alphama^{\sf det}) = O(\log W + \frac{\log n}{2})$, the $(\log n)/2$ term being of a lower order.
    There is an intermediate regime in which
    $n^{\Omega(1)}\leq W \leq \exp(O(n \log^2 n))$ and for which $
    -\log (1-\alphama^{\sf det}) = O( \sqrt{n\log n \log W})$. 

Note that the bound on $\alphabw$ from~\Cref{th-dubickas} is
incomparable to the one of~\Cref{coroZP96}:
depending on the value of $W$, none of the bounds dominates the other. Taking the best (smallest) of the two bounds, we arrive at~\Cref{table-regimes} in which we bound the Blackwell threshold depending on the different regimes for $\log(W)$ as a function of the number of states $n$ (up to terms of lower order). \todo{We need to update this remark with the new bound on $\alphabw$ (we don't have the max anymore). We may not need the notation $\alphama^{\sf det}$ anymore.}
\end{remark}

   %%%%% jgc --- still to address for final submission
\begin{remark}
  %   \todo[color=red!30]{SG: the bound is wrong (too small), I must correct it}
  % \todo{SG: I corrected the bound. FT results shows that if we have a collection $B$ of integer vectors $b$ of $\ell_1$ normless than $N-1$, in dimension $n$, and a rational vector $w$, then we can round the vector $w$ so that $\|\tilde{w}\|_\infty\leq 2^{4n^3}N^{n(n+2)}$. In our case, the dimension is $mn$, and the vectors $b$ are of the form $p1_I - q 1_J$ where $I,J$ are circuits of lengths $q$ and $p$, respectively, so $\|b\|_1 \leq n+m$. Leading to a bound of $2^{4(mn)^3}{(n+m)^{mn(mn+2)}}$}     
Using the preprocessing algorithm of Frank and Tardos~\cite{franktardos}, for any deterministic mean-payoff game with $n$ states of Min and $m$ states of Max, and arbitrary rational weights, we can construct a mean-payoff game with integer weights such that $W=2^{O(nm)^3}$ and which has the same optimal strategies.
For $n=m$, this estimate of $W$ is larger than the separation
%% between \Cref{rk-regimes} and~\Cref{table-regimes} should be read
%% with this estimate in mind.
%% In particular, if, to fix ideas, we set $n=m$,
%%  %   this estimate of $W$ reduces to $O((2n+1)^{n^2})$ which
%% we see that this estimate of $W$ is higher than the separation
order
$W=\exp(\Theta( n\log^2 n))$ between the ``intermediate''
and ``high $W$'' regimes in~\Cref{table-regimes}. So, all the regimes
in~\Cref{rk-regimes} and~\Cref{table-regimes}
are relevant.
%% priori significant from the perspective of computational
%% complexity.
\todo{does this remark need some updates}
%   \todo[inline]{SG: The proof of all this needs to be checked/detailed. To show this, I compare the mean-weights of two cycles and check the sign of the difference is unchanged by the reduction of Theorem 3.1, {\em ibid}. This bounds holds if we only to require to preserve
%     $-1$-optimal (mean-payoff optimal) policies. This raises the question of understanding how large should the rounding must be to preserve the set of Blackwell optimal policies (or higher sensitive  policies). In this way, we might get a threshold estimate independent of $W$ (by working out feasibility condition
%     for the lexicographic LP) but it will be huge.  }
%   \todo[inline]{SG: I was wondering whether we could show that the regime of ``large $W$'' is irrelevant, but the bound on $W$ gotten by diophantine approximaton does not allow us to get rid of this case. It may be the case that all regimes are interesting in terms of computational complexity.}
%   \todo[inline]{SG: The bound of \Cref{th:mahler}  gives in $-\log\epsilon_{\textrm{Blackwell}} \leq Cst + \log W + 2n\log 2 - \log n$ which surprisingly is worse than the bound of~\Cref{th-dubickas} in the high $W$ regime} 
\end{remark}
\todo{next table to be updated too (given new results with Mahler; also, can we find a way to indicate which bound dominates the other in each regime)}
  \begin{table}[htp]
        \caption{Bound for the Blackwell threshold $\alphabw$ of deterministic perfect-information SGs in different regimes of $\log(W)$ with respect to the number of states $n$.
        }
      \label{table-regimes}
      \begin{center}
    \setlength{\tabcolsep}{-4.5pt}
    \scriptsize
    \begin{tabular}{r@{\extracolsep{3pt}}rcccccccccl}
    \toprule
$      \log W$ :&
            &$\Theta(\log n)$&&$\Theta(\frac{n}{\log n})$&&
$\Theta(n)$&&$\Theta(n\log n)$&&$\Theta(n\log^2n)$&\\
      $      -\log (1-\alphabw)$: &$ \sqrt{n}\log n$& $ \mid$ & $ \sqrt{n\log n\log W}$& $\mid$ &$n\log 2$ & $\mid$ &$ \log W$ &$\mid$ &$ \sqrt{n\log n\log W}$ &$\mid$ &$  \log W$ \\
      \bottomrule
    \end{tabular}
      \end{center}

    \end{table}
\section{Results for perfect-information stochastic games}\label{SectionStochastic}
We now focus on the general case of perfect-information SGs. %where every pair of stationary strategies $(\sigma,\tau)$ induces a transition matrix $P^{\sigma,\tau} $ over the set of states. 
We start by studying the structure of the discounted value function $v^{\sigma,\tau}_{i}(\alpha)$ associated with a pair of stationary strategies $(\sigma,\tau)$ and a state $i$. It is well-known that this function is rational, i.e., it is the ratio of two polynomials (e.g., Lemma 10.1.3 in~\cite{puterman2014markov}). We define $\Delta(\alpha)$ as the numerator appearing in $v^{\sigma,\tau}_{i}(\alpha) - v^{\sigma',\tau'}_{i}(\alpha)$ (we refer to \Cref{app:proof stochastic} for the precise definition of $\Delta$).
Thus, to study the zeros of $\alpha \mapsto v^{\sigma,\tau}_{i}(\alpha) - v^{\sigma',\tau'}_{i}(\alpha)$ we can focus on studying the roots of $\Delta$. 
We first analyze the degree and the size of the coefficients of $\Delta$.
\begin{proposition}\label{prop:bound on Delta - stochastic}
%Let $(\sigma,\tau)$ and $(\sigma',\tau')$ be two pairs of stationary deterministic strategies. Then, 
Under Assumption \ref{assumption}, the polynomial $\Delta$ can be written as $\Delta(\alpha) = \sum_{k=0}^{2n-1} c_k \alpha^k$, where $|c_k| \leq 2 n W M^{2n-1} {2n - 1 \choose k}$ for all $k \in \{0,\dots , 2n-1\}$.
\end{proposition}
\Cref{prop:bound on Delta - stochastic} improves upon the corresponding results of \cite{grand2023reducing} for MDPs, which show that $|c_k| \leq 2 n W M^{2n} 4^n$ across all $k\in \{0,\ldots , 2n-1\}$. 
%To use the Lagrange bound, we consider the polynomial $\epsilon \mapsto \Delta(1-\epsilon)$, for which we show, using the identity $\sum_{k=q}^{m} {m\choose k} {k\choose q}= 2^{m-q}{m\choose q}$, that:
%\begin{equation}\label{eq:bound_gi}
%\Delta(1-\epsilon) = \sum_{i=0}^{2n-1} (-1)^i\epsilon^i g_i \makebox{ where }|g_i| \leq n W (2M)^{2n-1} 2^{1-i}{2n-1\choose i}, \forall \; i \in \{0,\ldots ,2n-1\}.
%\end{equation}

%\begin{proposition}\label{prop_threshold_sto}
% \todo{RK: please check}
%Let $j$ be the smallest index such that $g_j \neq 0$ in~\eqref{eq:bound_gi}. Then, the polynomial $\epsilon \mapsto \Delta(1-\epsilon)$ has no zeros in the interval $]0 ,\frac{2^{j-1}}{nW (2M)^{2n-1} {2n-1 \choose j+1}}[$.
%\end{proposition}
{\bf Results based on the Lagrange bound.}
%Using \Cref{prop_threshold_sto}, we can prove our main results for general perfect-information SGs using the Lagrange bound. We start with the bound on $\alphabw$.
Applying the Lagrange bound to the polynomial $\epsilon \mapsto \Delta(1-\epsilon)$, we prove the second part of \Cref{ThLagrange}.  We start with the bound on $\alphabw$.

\begin{corollary}\label{coro:bound blackwell factor - stochastic}
%Assume that $\Gamma$ is a perfect-information stochastic game with $n \in \N$ states, integer rewards with absolute values bounded by $W \in \N$, and $M \in \N$ is the common denominator appearing in the transition probabilities. 
Under \Cref{assumption}, we have
$\displaystyle
\alphabw \leq 1- \frac{2^{\lfloor \frac{2}{3} n \rfloor-2}}{nW (2M)^{2n-1} {2n-1 \choose \lfloor \frac{2}{3} n \rfloor}} \;.$
% \[
% \alphabw \leq 1- \frac{2^{\lfloor \frac{2}{3} n \rfloor-2}}{nW (2M)^{2n-1} {2n-1 \choose \lfloor \frac{2}{3} n \rfloor}} \; .
% \]
\end{corollary}

We now state the bound on $\alphad$ for non-deterministic SGs. Note that in the next result, we assume that the SG is unichain, i.e., that the Markov chain induced by any pair of stationary strategies is unichain. \todo{explain better why we need the unichain assumption. RK: I added the next sentence with this aim} This assumption is necessary to precisely connect the coefficients of $\Delta(\alpha)$ with the coefficients of the Laurent series expansion of $v^{\sigma,\tau}_{i}(\alpha) - v^{\sigma',\tau'}_{i}(\alpha)$.
\begin{corollary}\label{th:bound $\sens$-sensitive discount factor - stochastic}
%Assume that $\Gamma$ is a unichain perfect-information stochastic game with $n \in \N$ states, integer rewards with absolute values bounded by $W \in \N$, and $M \in \N$ is the common denominator appearing in the transition probabilities. 
If $\Gamma$ satisfies \Cref{assumption} and is unichain, the \sens-sensitive threshold $\alphad$ satisfies 
$\displaystyle
\alphad \leq 1-\frac{2^{\min \{ \sens+2 , \lfloor \frac{2}{3} n - 1 \rfloor \} - 1}}{nW (2M)^{2n-1} {2n-1 \choose \min \{ \sens+2 , \lfloor \frac{2}{3} n - 1 \rfloor \} + 1}} \;.$
% \[
% \alphad \leq 1-\frac{2^{\min \{ \sens+2 , \lfloor \frac{2}{3} n - 1 \rfloor \} - 1}}{nW (2M)^{2n-1} {2n-1 \choose \min \{ \sens+2 , \lfloor \frac{2}{3} n - 1 \rfloor \} + 1}} \; .
% \]
\end{corollary}
{\bf Results based on the Mahler bound.}
Next, we present the results for $\alphabw$ based on the Mahler bound. The proof follows the same lines as in the deterministic case (\Cref{th-dubickas}). %, and we also provide a statement with an explicit constant in Appendix \ref{app:proof stochastic}.
%%% statement without max
%\begin{theorem}\label{th-dubickas2} 
%If $\Gamma$ satisfies \Cref{assumption}, there exists an absolute constant $C>0$ such that
% \[\alphabw \leq 1 - \frac{1}{\exp\left(C \cdot \sqrt{n \log(n) \log \left(n W M^{2n-1} \sqrt{{2(2n - 1) \choose 2n -1}} \right) } \right)}.\]
%$\displaystyle \alphabw \leq 1 - e^{-C  \sqrt{n \log(n) \log \left(n W M^{2n-1} \sqrt{{2(2n - 1) \choose 2n -1}} \right) } } .$
%\end{theorem}
%%% previous statement with max
 \begin{theorem}\label{th-dubickas2} 
 If the game $\Gamma$ satisfies \Cref{assumption}, 
 for each $\epsilon > 0$ there exists a constant $a_\epsilon$ such that the Blackwell threshold $\alphabw$ satisfies
   $\alphabw \leq \alphama $,
   where $- \log (1-\alphama) =  \max \{ \beta_1 , \beta_2 \}$, $\beta_1 := ( \frac{\pi}{4} + \epsilon)\sqrt{(2n-1) \log(2n-1) \log \left( L \right) }$, $\beta_2 :=  a_\epsilon + \log \left( L \right)$, and $L := 2 n W M^{2n-1} \sqrt{{2(2n - 1) \choose 2n -1}}$.
% % \begin{align*}
% % & \beta_1 = ( \frac{\pi}{4} + \epsilon)\sqrt{(2n-1) \log(2n-1) \log \left( 2 n W M^{2n-1} \sqrt{{2(2n - 1) \choose 2n -1}} \right) }   \\
% % & \beta_2 =  a_\epsilon + \log \left( 2 n W M^{2n-1} \sqrt{{2(2n - 1) \choose 2n -1}} \right) \; .
% % \end{align*} 
% %\begin{align*}
% %& \beta_1 := ( \frac{\pi}{4} + \epsilon)\sqrt{(2n-1) \log(2n-1) \log \left( L \right) }, \beta_2 :=  a_\epsilon + \log \left( L \right), L := 2 n W M^{2n-1} \sqrt{{2(2n - 1) \choose 2n -1}}. 
% %\end{align*} 
 \end{theorem}
We end this section with a remark on the {\em multiplicity approach} highlighted in Theorems~\ref{th:bound on d} and~\ref{th:bound on k}. 
  \begin{remark}\label{rmk:bound on d - stochastic}
%Assume that $\Gamma$ is a perfect-information SG with $n \in \N$ states, integer rewards with absolute values bounded by $W \in \N$, and $M \in \N$ is the common denominator appearing in the transition probabilities. 
Under \Cref{assumption}, 
%the assumptions of \Cref{th:bound $\sens$-sensitive discount factor - stochastic},
using arguments similar to the ones used to prove \Cref{th:bound on k}, we get that every $\bar{d}^{\sto}(n,W)$-sensitive optimal strategy is Blackwell optimal, where 
 $\bar{d}^{\sto}(n,W)\coloneqq
  a \sqrt{(2n-1) \left(1+\log\left(2nWM^{2n-1} {{2n-1}\choose {n}}\right)\right)} $ for some $a>0$.
  % \[ \bar{d}^{\sto}(n,W)\coloneqq
  % c \sqrt{(2n-1) \cdot \left(1+\log\left(2nWM^{2n-1} {{2n-1}\choose {n}}\right)\right)} \enspace ,
  % \]
However, $d=n-2$ is enough~\cite{Veinott1969}, and $n-2\ll \bar{d}^{\sto}(n,W)$. This shows that while the multiplicity approach may be useful in the deterministic case (in the sense that there are some regimes of $n$ and $W$ where $\bar{d}^{\det}(n,W) < n-2$), it does not yield any improvement over existing bounds in the stochastic case, highlighting the strength of our new results based on the Lagrange and Mahler bounds.
    \end{remark}
\section{Discussion}\label{sec:discussion}
We obtain bounds on the Blackwell threshold $\alphabw$ and on the $d$-sensitive threshold $\alphad$. We improve the existing bounds on $-\log(1-\alphabw)$ by a factor $\Omega(n)$ (compared to~\cite{andersson2009complexity} for SGs and \cite{grand2023reducing} for MDPs), and we provide the first bound on $\alphad$ beyond the case $d=-1$ in deterministic games. Blackwell and mean-payoff optimal strategies have received some attention in reinforcement learning and SGs in recent years, and our bounds control the complexity of the main method to compute Blackwell and $\sens$-sensitive optimal strategies for SGs, by choosing $\alpha > \alphabw$ (or $\alpha > \alphad$) in algorithms for solving discounted SGs.
%, an active area of research.
A crucial advantage of this approach is that our bounds can be combined with {\em any} progress in solving discounted SGs.
However, to the best of our knowledge, all algorithms whose complexity depends on the discount factor $\alpha$ scale as $\tilde{O}((1-\alpha)^{k})$ for some negative $k$ (e.g. $k=-1$ for strategy iteration~\cite{ye2011simplex,hansen2013strategy,akian2013policy},
%% $k=-2$ for two-player policy iteration \`a la Hoffman-Karp~\cite{hansen2013strategy,akian2013policy}
or $k=-3$ for sampling-based methods~\cite{sidford2020solving}). Our bounds on $(1-\alphabw)^{-1}$ involve some terms that grow superpolynomially in $n$ (see Tables~\ref{tab:bound deterministic SGs} and~\ref{tab:bound general SGs}). Obtaining stronger bounds for $\alphabw$ is an important future research direction. \todo{we need to better sell our results}
% Recall that our goal is to compute Blackwell optimal and mean-payoff optimal strategies by solving a discounted SG with discount factor close enough to $1$. The number of outer loops of Strategy Iteration is in $O(\frac{nm}{1-\alpha}\log(1/(1-\alpha))$~\cite{akian2013policy}), and each iteration of the outer loop requires solving an MDP, which can be done in $O(n^4 m^4\log(n/(1-\alpha))$~\cite{ye2005new}, where $m$ is the number of actions available for each player at each state. 

\noindent
    {\em Open questions.}
    %Our results uncover various interesting research directions.
    %% It would be interesting to understand what happens in more general settings (e.g. for continuous states or action sets).
    It would be interesting to obtain bounds for $\alphad$ for general (non-deterministic) perfect-information SGs, {\em without} the unichain assumption. Additionally, our work provides several different upper bounds for $\alphabw$, all of which may be exponentially close to $1$. It is essential to understand if this is a limitation of our line of analysis - for instance, does there exist stronger separation results between a root of a polynomial $P$ and $1$ for the specific polynomials $\Delta$ appearing in our proof? -, or if this is an inherent difficulty of the problem. Obtaining lower bounds for $\alphabw$ and $\alphad$ in all generality is an important next step.
    %% We also note that a {\em smoothed analysis} of $\alphabw$ has been studied in \cite{loff2024smoothed} for the case of ergodic deterministic SGs, where the authors show that with probability $-\log(1-\alphabw)$ remains small (see Theorem~1.8 and Theorem~1.9 in \cite{loff2024smoothed} for more precise results). Extending this smoothed analysis to the case of SGs is also an interesting question.
% jgc --- replaced "promising" by interesting"
% \todo[color=red!30]{SG: I am not sure whether we should say ``promising'', not clear it works, perhaps say ``interesting question''?}

\bibliographystyle{alpha}
\bibliography{tropical}
\appendix
\section{Bounds from previous work}\label{app:bounds from past work}
{\bf Bounds from \cite{andersson2009complexity}.} The authors in \cite{andersson2009complexity} focus on perfect-information SGs. Lemma~1 of~\cite{andersson2009complexity} shows that $\alphabw \leq 1- 2 (n!)^2 4^{n} \max\{M,W\}^{2n^2}$. We then apply the classical bound $n! = O\left(\sqrt{n} \left(\frac{n}{e}\right)^{n} \exp\left(\frac{1}{12 n}\right) \right)$ to obtain the bound presented in Tables~\ref{tab:bound deterministic SGs} and~\ref{tab:bound general SGs}.

\noindent 
{\bf Bounds from \cite{grand2023reducing}.} The authors of \cite{grand2023reducing} focus on MDPs. Their main bound for $\alphabw$ is given in Theorem~4.4, and their bound can be simplified to $-\log(1-\alphabw) = O\left(n \log(W) + n^2 (1+\log(M)\right)$, see the calculation in Appendix E of \cite{grand2023reducing}.

\noindent 
{\bf Bounds from \cite{mukherjee2025howard}.} The bound for $-\log(1-\alphabw)$ is given in Section~4.4 of \cite{mukherjee2025howard}, with the notation $b$ for our term $\log(W)$.

\section{Proofs for Section \ref{SectionDeterm}}\label{app:proof deterministic}
%{\bf Proof of Proposition \ref{prop_threshold_det}}
%\Cref{prop_threshold_det} is a corollary of the following technical lemma.

In this section we provide the proofs of the results of \Cref{SectionDeterm}. We first prove \Cref{lem:degree coeff Delta - deterministic}.
%For $q \in \N$, we set
%\[\q[q]:=1+\alpha+ \dots + \alpha^{q-1}.\]
\begin{proof}[Proof of \Cref{lem:degree coeff Delta - deterministic}]
%   we can choose $\pi$ and $\gamma$ such that $p+q\leq n$ and $p\leq n-1$. 
%    Let us consider the polynomial in $\alpha$
Note that the polynomial $\Delta$ defined in~\eqref{eq:def_Delta} satisfies
\begin{align*}
%  \Delta & := (1-\alpha^q)(1-\alpha^{q'}) (v_i^{\sigma, \tau} (\alpha) - v_i^{\sigma',\tau'} (\alpha) )\\
\Delta (\alpha )  &= (1-\alpha^q)(1-\alpha^{q'})(\<r,\pi>-\<r,\pi'>)+ (1-\alpha^{q'})\alpha^p \<r,\gamma> - (1-\alpha^{q})\alpha^{p'} \<r,\gamma'>
%\\     &= (1 -\alpha)^2 \q[q] \q[q'](\<r,\pi>-\<r,\pi'>)+(1 -\alpha) \left( \q[q'] \alpha^p \<r,\gamma> - \q[q] \alpha^{p'} \<r,\gamma'> \right) 
\; .
\end{align*}
Thus, if the absolute value of the instantaneous rewards are bounded by $W$, the coefficients of the polynomial of \Cref{lem:degree coeff Delta - deterministic} satisfy 
%We can rewrite $\Delta$ as
%\[
%\Delta = \sum_{k=0}^K a_k \alpha^k
%\]
%in which $|a_k|\leq 12 W$, and
 $|a_k|\leq 12 W$ for all $k$. In addition, we have
\[ K = \max \{ q+q'+\max\{p,p'\}-1,q'+p+q-1, q+p'+q'-1\}\leq 2n-1\enspace , 
\]
because we can choose $\pi$, $\gamma$, $\pi'$ and  $\gamma'$ elementary (i.e., such that in the corresponding sequences of states, no state appears twice, except the initial and last state in the case of circuits), and so $p+q\leq n$, $p\leq n-1$, $p'+q'\leq n$ and $p'\leq n-1$. 
\end{proof}

\subsection{Proofs for Section \ref{sec:determ Lagrange}}
We now provide the intermediate results that we need to prove Theorem \ref{th:bound $\sens$-sensitive discount factor - deterministic}.
In what follows, we denote by $H(P)$ the height of the polynomial $P=\sum^{d}_{k=0} c_k x^k$ defined as $H(P) \coloneqq  \max_{k \in \{ 0, \ldots , d\}} | c_k |$. 
\begin{lemma}\label{lemmaTechnical2} 
Let $P =\sum^{d}_{k=0} c_k x^k$ be a polynomial with integer coefficients and $Q(y)=\sum^{d}_{k=0} c'_k y^k$ be the polynomial which is obtained making the change of variable $x = 1 - y$ in $P(x)$. Suppose that, %$c'_0 = \ldots = c'_{j-1} = 0$ and $c'_j \neq 0$ for some $j \geq 1$. 
for some $j\in \{0, \ldots , d\}$, we have $c'_j \neq 0$ and $c'_i = 0$ for all $i < j$. 
Then, the polynomial $Q(y)$ has no zeros in the interval $]0 , \frac{1}{2 H(P){d+1\choose j+2}}[$.
\end{lemma}

\begin{proof}
We first bound the magnitude of the coefficients of $Q$ in terms of $H(P)$. Since $Q(y) = P(1-y)$, we have $c'_i = (-1)^i \sum_{k=i}^d c_k {k\choose i}$ for $i \in \{0, \ldots , d\}$. Now we use $\sum_{k=i}^d {k \choose i} = {d+1 \choose i+1} $ to obtain that $|c'_i| \leq H(P) {d+1 \choose i+1}$ for all $i \in \{0, \ldots , d\}$.

Note that $|c'_j|\geq 1$ since $c'_j \neq 0$ and $c'_j \in \Z$. %Then, for any $i \in \{ j+1, \ldots , d\}$, we have 
%\begin{equation}\label{eq:intermed eq 0}
%\left(\frac{|c'_{j}|}{|c'_i|}\right)^{\frac{1}{i-j}} \geq \frac{|c'_{j}|^{\frac{1}{i-j}}}{H(P)^\frac{1}{i-j}     {d+1 \choose i+1}^\frac{1}{i-j}}.
   % \geq \frac{1}{H(P) {d +1 \choose j+2}} 
%\end{equation}
Besides, note that for $m \geq 3$ and $j\leq m-1$, we have 
\begin{equation}\label{eq:aux} 
{m \choose i} \leq {m \choose j+1}^{i-j}
\end{equation}
for all $i \in \{j+1, \ldots , m\}$. Indeed, the cases $i = j+1$ and $i = m$ are trivial, so we may assume $j < m -1$ (because $j = m -1$ implies $i = m $). Suppose that ${m \choose i} \leq {m \choose j+1}^{i-j}$ for some $i\leq m-1$. Then,  \begin{align*}
{m \choose i+1} = \frac{m-i}{i+1}{m \choose i}  \leq \frac{m}{2} {m \choose j+1}^{i-j} \leq {m \choose j+1}^{i + 1 -j}\enspace ,
\end{align*}
because $\frac{m}{2} \leq 
\frac{m}{2}\frac{m-1}{j+1}\cdots \frac{m-j+1}{3}\frac{m-j}{1}= \frac{m}{j+1}\frac{m-1}{j}\cdots \frac{m-j+1}{2}\frac{m-j}{1}= {m \choose j+1}$. 

We are now ready to apply the Lagrange bound (\Cref{th:Lagrange}). 
Using \eqref{eq:aux} we get
\[ 
\left(\frac{|c'_{j}|}{|c'_i|}\right)^{\frac{1}{i-j}} \geq  \frac{|c'_{j}|^{\frac{1}{i-j}}}{H(P)^\frac{1}{i-j}     {d+1 \choose i+1}^\frac{1}{i-j}}
   \geq \frac{1}{H(P) {d +1 \choose j+2}} 
 \]
for any $i \in \{ j+1, \ldots , d\}$. %We conclude that $\gamma \geq \frac{1}{2 H(P) {d +1 \choose j+2}}$, which completes the proof. 
The lemma now follows from \Cref{th:Lagrange}.
\end{proof}

Note that the polynomial $\epsilon \mapsto \Delta(1-\epsilon)$ can be rewritten as
\begin{align}\label{DeltaEpsilon}
  \Delta(1-\epsilon)  
  % =   \sum_{k=0}^K a_k (1-\epsilon)^k
   % = \sum_{k=0}^K a_k (\sum_{i=0}^k {k\choose i}(-1)^i\epsilon^i)
   = \sum_{i=0}^K (-1)^i\epsilon^i b_i \enspace 
   % \text{ where } b_i &= \sum_{k=i}^K  a_k {k\choose i}\enspace .
\end{align}
where $b_i = \sum_{k=i}^K  a_k {k\choose i}$. 
Then, the next proposition is a direct consequence of Lemmas~\ref{lemmaTechnical2} and~\ref{lem:degree coeff Delta - deterministic}.

\begin{proposition}\label{prop_threshold_det}
Suppose that $b_0 = \ldots = b_{j-1} = 0$ and $b_j \neq 0$ for some $j \geq 1$ in~\eqref{DeltaEpsilon}. Then, the polynomial $\epsilon \mapsto \Delta(1-\epsilon)$ has no zeros in the interval $]0 , \frac{1}{24 W{K+1\choose j+2}}[$.
\end{proposition} 

\begin{proof}[Proof of \Cref{th:bound $\sens$-sensitive discount factor - deterministic}]
% \todoi{RK: this proof uses the definition of $\sens$-sensitive optimality of the note, which uses $\liminf$ and $\limsup$ for players $\Max $ and $\Min$ respectively. Since that definition differs from the one used here (no distinction between the players is made), one of the two should be chosen.}
% Jgc: I changed the proof here (and in the appendix for the unichain case) to the case of symmetric definition for d-sensitivity
Let $\alpha'$ be such that $1-\frac{1}{24 W{2n\choose {\min \{ \sens+4, n\} }}} < \alpha' < 1$, and let $(\sigma^*,\tau^*)$ be a pair of discount optimal strategies for the discount factor $\alpha'$. To prove the theorem, it is enough to show that $(\sigma^*,\tau^*)$ is a pair of $\sens$-sensitive optimal strategies.

On the contrary, suppose that $(\sigma^*,\tau^*)$ is not a pair of $\sens$-sensitive optimal strategies. Then, either there exist a strategy $\tau$ of player $\Max$ and a state $i$ such that
\begin{equation}\label{Eq1Th1}
\lim_{\alpha\to 1^-} \; (1-\alpha)^{-\sens}(v_i^{\sigma^*, \tau^*} (\alpha) - v_i^{\sigma^*,\tau} (\alpha)) < 0 \; ,
\end{equation}
or there exist a strategy $\sigma$ of player $\Min$ and a state $i$ such that 
\begin{equation}\label{Eq2Th1}
\lim_{\alpha\to 1^-} \; (1-\alpha)^{-\sens}(v_i^{\sigma^*, \tau^*} (\alpha) - v_i^{\sigma,\tau^*} (\alpha)) > 0 \;.
\end{equation}

In the first place, assume that~\eqref{Eq1Th1} holds. Let $d'$ be the smallest value satisfying 
\begin{equation}\label{Eq1Th1-general}
\lim_{\alpha\to 1^-} \; (1-\alpha)^{-\sens'}(v_i^{\sigma^*, \tau^*} (\alpha) - v_i^{\sigma^*,\tau} (\alpha)) < 0 \; .
\end{equation}
By~\eqref{Eq1Th1}, it follows that $d' \leq d$ and that $
\lim_{\alpha\to 1^-} \; (1-\alpha)^{-\sens''}(v_i^{\sigma^*, \tau^*} (\alpha) - v_i^{\sigma^*,\tau} (\alpha)) = 0$ for all $d'' < d'$. 

If for $q \in \N$ we set $\q[q]:=1+\alpha+ \dots + \alpha^{q-1}$,  then the polynomial $\Delta$ defined in~\eqref{eq:def_Delta} satisfies
\begin{equation}\label{eq:property_Delta}
\Delta(\alpha) = (1-\alpha^q)(1-\alpha^{q'}) (v_i^{\sigma, \tau} (\alpha) - v_i^{\sigma',\tau'} (\alpha) ) = (1 -\alpha)^2 \q[q] \q[q'] (v_i^{\sigma, \tau} (\alpha) - v_i^{\sigma',\tau'} (\alpha) )  \; .
\end{equation}
Therefore, for all $\sens''$ we have
    \[
        \lim_{\alpha\to 1^-} \; (1-\alpha)^{-\sens''}(v_i^{\sigma^*, \tau^*} (\alpha) - v_i^{\sigma^*,\tau} (\alpha)) = 0 \iff 
\lim_{\alpha\to 1^-} \; (1-\alpha)^{-(\sens'' +2)} \Delta(\alpha) = 0 \; .
    \]
     We conclude that $b_0=\ldots = b_{\sens'+1}=0$ and $b_{\sens'+2}\neq 0$ if we represent the polynomial $\Delta(1-\epsilon)$ as in~\eqref{DeltaEpsilon}. Thus, by~\eqref{Eq1Th1} and \Cref{prop_threshold_det},  it follows that $v_i^{\sigma^*, \tau^*} (\alpha) - v_i^{\sigma^*,\tau} (\alpha) < 0$ for $1-\frac{1}{24 W{K+1\choose {\sens' + 4}} } < \alpha < 1$, and therefore this remains true for $1-\frac{1}{24 W{2n\choose {\min \{ \sens + 4 , n\}}}} < \alpha < 1$ because $K+1 \leq 2n$ and $\sens' \leq \sens$. Since $1-\frac{1}{24 W{2n\choose {\min \{\sens+4, n\}}}} < \alpha' < 1$, in particular we have $v_i^{\sigma^*, \tau^*} (\alpha') - v_i^{\sigma^*,\tau} (\alpha') < 0$, contradicting the fact that $(\sigma^*,\tau^*)$ is a pair of discount optimal strategies for $\alpha'$.

On the other hand, if~\eqref{Eq2Th1} is satisfied, using symmetric arguments we also arrive to a contradiction. This completes the proof.
\end{proof}

\begin{proof}[Proof of \Cref{coroZP96}]
 % \todo{RK: please check}
 The first inequality of this corollary readily follows from \Cref{th:bound $\sens$-sensitive discount factor - deterministic}.
  The second inequality follows from \Cref{prop_threshold_det}, the fact that $K \leq 2n-1$ and that the function $j \mapsto \frac{1}{24 W  {2n \choose j+2}}$ is convex and achieves its minimum at $j=n-2$.
 \end{proof}
 We conclude this section with the proof of \Cref{th:bound on k}.
\begin{proof}[Proof of \Cref{th:bound on k}]
  % \todo{RK: please check!}
Theorem 2.1 of \cite{BEK1999} shows that if $P=\sum_{k=0}^d c_k x^k$ is a non-zero polynomial such that $\max_{k \in \{ 0, \ldots , d\}} | c_k | \leq 1$, then the multiplicity of $1$ as a root of $P$ cannot exceed $a \sqrt{d(1-\log |c_0|)}$, where $a>0$ is an absolute constant. Applying this result to the polynomial $\frac{\Delta(\alpha)}{12W}$, we conclude that the multiplicity of $1$ as a root of $\Delta(\alpha)$ is at most $a \sqrt{(2n-1)(1+\log 12W)}$. Then,  by~\eqref{eq:property_Delta} we can write $v^{\sigma,\tau}_{i}(\alpha)-v^{\sigma,\tau'}_{i}(\alpha)$ as $(1 - \alpha )^k Q(\alpha)$, where $k \leq \bar{d}^{\de}(n,W) \coloneq a \sqrt{(2n-1)(1+\log 12W)}-2$ and $Q (\alpha)$ is a continuous function satisfying $Q (1) \neq 0$. 

If the pair of strategies $(\sigma,\tau)$ is $\bar{d}^{\de}(n,W)$-sensitive optimal, we have 
\[
0 \leq \lim_{\alpha\to 1^-} \; (1-\alpha)^{-\bar{d}^{\de}(n,W)}\left(v^{\sigma,\tau}_{i}(\alpha)-v^{\sigma,\tau'}_{i}(\alpha) \right)  = \lim_{\alpha\to 1^-} (1-\alpha)^{k-\bar{d}^{\de}(n,W)} Q(\alpha) \; ,
\]
and so $Q(1) > 0$. It follows that  
\[
\lim_{\alpha\to 1^-} \; (1-\alpha)^{d}\left(v^{\sigma,\tau}_{i}(\alpha)-v^{\sigma,\tau'}_{i}(\alpha) \right)  = \lim_{\alpha\to 1^-} (1-\alpha)^{k-d} Q(\alpha) \geq 0 
\]
for any $d$. Using similar arguments it is possible to show also that 
\[
\lim_{\alpha\to 1^-} \; (1-\alpha)^{d}\left(v^{\sigma',\tau}_{i}(\alpha)-v^{\sigma,\tau}_{i}(\alpha) \right)   \geq 0 
\]
for any $d$ and any strategy $\sigma'$ of player Min.
We conclude that $(\sigma,\tau)$ is $\sens$-sensitive optimal for any $\sens$, and so also Blackwell optimal.
%%%%% jgc --- still to address for final submission
% \todo{RK: I believe that this last conclusion was only proved for MDPs, but also that it should not be very difficult to prove it for the games we consider. Should we say something?}
  \end{proof}
\subsection{Proof for Section \ref{sec:determ mahler measures}}\label{app:detailed Dubickas}

% The next is the proof of the statement without the max 
%We start by reproducing verbatim Theorem~1 of~\cite{Dubickas1995}.
%\begin{theorem}%[{\cite[Theorem~1]{Dubickas1995}]
    \label{th:mahler-1}
\begin{proof}[Proof of \Cref{th-dubickas}]
Given $\epsilon > 0$, let $D_\epsilon$ be the constant provided by \Cref{th:mahler}.
Let $z$ be any real root of the polynomial $\Delta$ different from  $1$, $d$ be its degree and $P =\sum^{d}_{k=0} c_k x^k$ be its minimal polynomial. Since $z$ is a root of $\Delta$, it follows that $P$ divides $\Delta$, and so we have $d \leq 2n -1$ and $M(P) \leq M(\Delta)$ (see Section~1.3 of \cite{cerlienco1987computing}). Besides, note that $M(\Delta)\leq 12W \sqrt{2n}$ by Landau's bound~\cite{Landau1905}.
% \todo{RK: this sentence is critical and needs to be checked.}
If $d > D_\epsilon$, then~\eqref{dubickas} holds, and so we have
\begin{align*}
|z - 1 | & > e^{-(\pi/4 + \epsilon) \sqrt{d \log d  \log M(P)}} \geq e^{-(\pi/4 + \epsilon) \sqrt{(2n-1) \log (2n-1) \log M(\Delta)}} \\
& \geq e^{-(\pi/4 + \epsilon) \sqrt{(2n-1) \log (2n-1)  \log (12 W \sqrt{2n}) }} \; . 
\end{align*}
Assume now that $d \leq D_\epsilon$. By Lemma~\ref{lemmaTechnical2}, it follows that
$|z -1 | \geq \frac{1}{2 H(P) {d+1\choose \lceil \frac{d}{2}\rceil}}$. Then, since $H(P) \leq 2^d M(P)$, we have
\[
|z -1 | \geq \frac{1}{2^{d+1} M(P) {d+1\choose \lceil \frac{d}{2}\rceil }}
\geq \frac{1}{2^{D_\epsilon +1} M(\Delta) {D_\epsilon +1\choose \lceil \frac{D_\epsilon}{2}\rceil }} \geq \frac{1}{2^{D_\epsilon +1} {D_\epsilon +1\choose \lceil \frac{D_\epsilon}{2}\rceil } 12 W \sqrt{2n}} \enspace .
\]
Setting $a_\epsilon = \log\left(2^{D_\epsilon +1} {D_\epsilon +1\choose \lceil \frac{D_\epsilon}{2}\rceil }\right)$,
we conclude that $\Delta$ has no zeros in the interval  $]\alphama^{\sf det}, 1[$. %and so the difference $v_i^{\sigma, \tau} (\alpha) - v_i^{\sigma',\tau'} (\alpha)$ has constant sign in this interval for any pairs of strategies $(\sigma, \tau)$ and $(\sigma', \tau')$. This concludes the proof of Theorem \ref{th-dubickas-1}.
\end{proof}
\section{Proof for \Cref{SectionStochastic}}\label{app:proof stochastic}
\subsection{Proof of \Cref{prop:bound on Delta - stochastic}}
We now detail the proof of \Cref{prop:bound on Delta - stochastic}. 

We first provide the exact formula for the polynomial $\Delta(\alpha)$ considered in \Cref{SectionStochastic}.  Given a pair of stationary strategies $(\sigma,\tau)$, let $P^{\sigma,\tau} $ be the transition matrix and $r^{\sigma,\tau}$ be the vector of instantaneous rewards induced by $(\sigma,\tau)$, and let us define $Q^{\sigma,\tau} \coloneq M P^{\sigma,\tau}$ and $D^{\sigma,\tau}(\alpha) \coloneq \det \left(M I - \alpha Q^{\sigma,\tau}\right)$. Since it is known that $\left(v^{\sigma,\tau}_{i}(\alpha)\right)_{i \in [n]} = \left(I - \alpha P^{\sigma,\tau}\right)^{-1}r^{\sigma,\tau}$, using Cramer's formula for the inverse of a matrix and Laplace's cofactor extension, it follows that $v^{\sigma,\tau}_{i}(\alpha) - v^{\sigma',\tau'}_{i}(\alpha) = \frac{M \Delta(\alpha)}{D^{\sigma,\tau}(\alpha)D^{\sigma',\tau'}(\alpha)}$ for any state $i$, where
\begin{equation}\label{eq:def_Delta_stoch}
\Delta(\alpha) \coloneq  D^{\sigma',\tau'} (\alpha)(\sum_{j=1}^{n} \cof_{ji} (M I - \alpha Q^{\sigma,\tau}) r_j^{\sigma,\tau} ) 
 -  D^{\sigma,\tau}(\alpha) (\sum_{j=1}^{n} \cof_{ji} (M I - \alpha Q^{\sigma',\tau'}) r_j^{\sigma',\tau'}) . %\label{DeltaP} 
\end{equation}

We now proceed to bound the degree and the coefficients of $\Delta$. In the next two propositions, we assume that $\frac{Q}{M}$ is a row-stochastic matrix, with $Q \in \N^{n \times n}$. 
\begin{proposition}\label{Prop_det_stoch_bound}
The polynomial $\det (M I - \alpha Q)$ is of the form $\sum_{k=0}^n a_k \alpha^k$, where $|a_k| \leq {n \choose k} M^n$. 
\end{proposition}

\begin{proof}
We have 
\begin{equation}\label{det_stoch_bound}
\det (M I - \alpha Q) = \sum_{k=0}^n (-\alpha)^k M^{n - k} \tr (\comp_k (Q)) \enspace ,
\end{equation}
where $\comp_k (Q)$ is the $k$-th compound matrix of $Q$. Since the entries of $\comp_k (Q)$ are minors of $Q$ of size $k \times k$, and $\sum_{j=1}^n Q_{ij} = M$ for each $i= 1,\ldots ,n$, we conclude that the absolute value of all the entries of $\comp_k (Q)$ are less than or equal to $M^k$. The result now follows from~\eqref{det_stoch_bound} and the fact that $\comp_k (Q)$ is of size ${n \choose k} \times {n \choose k}$.
\end{proof}
It is worth noting that the result in the previous proposition is tight when $Q = M I$.
\begin{proposition}\label{Prop_cofactor_bound}
For each $i,j \in \{1, \ldots , n\}$, the $(i,j)$ cofactor of the matrix $M I - \alpha Q$ is of the form $\sum_{k=0}^{n-1} a_k \alpha^k$, where $|a_k| \leq {n-1 \choose k} M^{n-1}$. 
\end{proposition}

\begin{proof}
The $(i,j)$ cofactor of the matrix $M I - \alpha Q$ is given by
\begin{equation}\label{cofactor_stoch_bound}
(-1)^{i+j} \det (M J - \alpha R) = (-1)^{i+j} \sum_{k=0}^{n-1} (-\alpha)^k  M^{n - 1 - k}  \tr ((\adj_k (J)) (\comp_k (R))) \enspace ,
\end{equation}
where $J$ and $R$ are the $(i,j)$ sub-matrices of $I$ and $Q$ respectively, $\adj_k (J)$ is the $k$-th higher adjugate matrix of $J$, and $\comp_k (R)$ is the $k$-th compound matrix of $R$. Since $\adj_k (J)$ has just one non-zero entry per row, of absolute value one, as in the proof of Proposition~\ref{Prop_det_stoch_bound} we conclude that the absolute value of all the entries of $(\adj_k (J))(\comp_k (R))$ are less than or equal to $M^k$. The result now follows from~\eqref{cofactor_stoch_bound} and the fact that $(\adj_k (J)) (\comp_k (R))$ is of size ${n-1 \choose k} \times {n-1 \choose k}$.
\end{proof}

By Propositions~\ref{Prop_det_stoch_bound} and~\ref{Prop_cofactor_bound}, both terms appearing in the polynomial $\Delta$ defined in \eqref{eq:def_Delta_stoch} are of the form $\sum_{k=0}^{2n-1} b_k \alpha^k$, where
% \todo{RK : In comparison to the notes, the bound that I obtain here has a factor $n$ instead of $n-1$.} 
\begin{align*}
|b_k| & \leq n W M^{2n-1} \left(\sum_{s+l=k,s\leq n-1,l\leq n} {n-1 \choose s} {n \choose l}\right)\\ 
& = n W M^{2n-1} \left(\sum_{0 \leq k-l\leq n-1,l\leq n} {n -1 \choose k-l} {n \choose l}\right)\\ 
& \leq n W M^{2n-1} \left(\sum_{l=0}^{k} {n -1 \choose k-l} {n \choose l}\right) = n W M^{2n-1} {2n - 1 \choose k} \enspace ,
\end{align*}
by Vandermonde's Identity. Therefore, we can rewrite $\Delta$ as $\sum_{k=0}^{2n - 1} c_k \alpha^k$, where $|c_k| \leq 2 n W M^{2n-1} {2n - 1 \choose k}$. This concludes the proof of \Cref{prop:bound on Delta - stochastic}.
%\subsection{Proof of \Cref{th:bound $\sens$-sensitive discount factor - stochastic}}
% \todoi{RK: this proof uses the definition of $\sens$-sensitive optimality of the note, which uses $\liminf$ and $\limsup$ for players $\Max $ and $\Min$ respectively. Since that definition differs from the one used here (no distinction between the players is made), one of the two should be chosen.}
\subsection{Other proofs for Section \ref{SectionStochastic}}

To use the Lagrange bound, we consider the polynomial $\epsilon \mapsto \Delta(1-\epsilon)$, which can be rewritten as
\begin{equation}\label{eq:bound_gi}
\Delta(1-\epsilon) = \sum_{i=0}^{2n-1} (-1)^i\epsilon^i g_i 
\end{equation}
where $g_i = \sum_{k=i}^{2n-1} c_k {k\choose i}$. Then,  using \Cref{prop:bound on Delta - stochastic} and the identity $\sum_{k=q}^{m} {m\choose k} {k\choose q}= 2^{m-q}{m\choose q}$, we have 
\begin{equation}\label{boundg}
|g_i| \leq \sum_{k=i}^{2n-1} 2 n W M^{2n-1} {2n - 1 \choose k} {k\choose i} 
%\leq 2n W M^{2n-1} \sum_{k=i}^{2n} {2n \choose k} {k\choose i} 
=  n W (2M)^{2n-1} 2^{1-i}{2n-1\choose i}
\end{equation}
for all $i \in \{0,\ldots ,2n-1\}$.

\begin{proposition}\label{prop_threshold_sto}
% \todo{RK: please check}
Let $j$ be the smallest index such that $g_j \neq 0$ in~\eqref{eq:bound_gi}. Then, the polynomial $\epsilon \mapsto \Delta(1-\epsilon)$ has no zeros in the interval $]0 ,\frac{2^{j-1}}{nW (2M)^{2n-1} {2n-1 \choose j+1}}[$.
\end{proposition}

\begin{proof}%[Proof of Proposition \ref{prop_threshold_sto}]
As in the proof of \Cref{lemmaTechnical2}, we apply the Lagrange bound (\Cref{th:Lagrange}). 
Note that $|g_j|\geq 1$ since $g_j \neq 0$ and $g_j \in \Z$. 
Then, for any $i\in \{ j+1, \ldots , 2n-1\}$, we have 
  \begin{align*}
   \left(\frac{|g_{j}|}{|g_i|}\right)^{\frac{1}{i-j}} &\geq \frac{|g_{j}|^{\frac{1}{i-j}}}{(nW (2M)^{2n-1}2^{1-i})^\frac{1}{i-j}
     {2n-1 \choose i}^\frac{1}{i-j}}
   \geq \frac{2^{i-1}}{nW (2M)^{2n-1} {2n-1 \choose j+1}} % \\ &\geq \frac{2^{j}}{nW (2M)^{2n-1} {2n-1 \choose j+1}} 
 \end{align*}
by~\eqref{boundg} and~\eqref{eq:aux}. 
The proposition now follows from \Cref{th:Lagrange}.
\end{proof}

\begin{proof}[Proof of \Cref{coro:bound blackwell factor - stochastic}]
Note that the function $j \mapsto \frac{2^{j-1}}{nW (2M)^{2n-1} {2n-1 \choose j+1}}$ is convex and attains its minimum at $j = \lfloor \frac{2}{3} n - 1 \rfloor$. Then, by \Cref{prop_threshold_sto}, we conclude that 
no function $\alpha \mapsto v^{\sigma,\tau}_{i}(\alpha) - v^{\sigma',\tau'}_{i}(\alpha)$ has zeros in the interval $]1-\frac{2^{\lfloor \frac{2}{3} n \rfloor-2}}{nW (2M)^{2n-1} {2n-1 \choose \lfloor \frac{2}{3} n \rfloor}}, 1[$. The corollary now follows from the discussion in \Cref{sec:stochastic games}.
\end{proof}

\begin{proof}[Proof of \Cref{th:bound $\sens$-sensitive discount factor - stochastic}]
Let $(\sigma^*,\tau^*)$ be a pair of discount optimal strategies for the discount factor
$\alpha'$, where $\alpha'$ satisfies $1-\frac{2^{\min \{ \sens+2 , \lfloor \frac{2}{3} n - 1 \rfloor \} - 1}}{nW (2M)^{2n-1} {2n-1 \choose \min \{ \sens+2 , \lfloor \frac{2}{3} n - 1 \rfloor \} + 1}} < \alpha' < 1$. 

As in the proof for the deterministic case, in the first place assume that there exist a strategy $\tau$ and a state $i$ such that~\eqref{Eq1Th1} is satisfied, and let $d'$ be the smallest value satisfying~\eqref{Eq1Th1-general}. %$\lim_{\alpha\to 1^-} \; (1-\alpha)^{-\sens'}(v_i^{\sigma^*, \tau^*} (\alpha) - v_i^{\sigma^*,\tau} (\alpha)) < 0$. 
By~\eqref{Eq1Th1}, it follows that $d' \leq d$ and that $
\lim_{\alpha\to 1^-} \; (1-\alpha)^{-\sens''}(v_i^{\sigma^*, \tau^*} (\alpha) - v_i^{\sigma^*,\tau} (\alpha)) = 0$ for all $d'' < d'$. 

If the game is unichain, there exist two polynomials $p(\alpha)$ and $q(\alpha)$ such that $p(1)\neq 0$, $q(1)\neq 0$ and $\Delta(\alpha) = (1-\alpha)^2 
 p(\alpha) q(\alpha) (v_i^{\sigma^*, \tau^*} (\alpha) - v_i^{\sigma^*,\tau} (\alpha)) $,
 % \todo{RK: should we justify this?}
 where $\Delta$ is the polynomial defined in \eqref{eq:def_Delta_stoch} considering the pairs of strategies $(\sigma^*, \tau^*)$ and $(\sigma^*, \tau)$. Then, we have
    \[
        \lim_{\alpha\to 1^-} \; (1-\alpha)^{-\sens''}(v_i^{\sigma^*, \tau^*} (\alpha) - v_i^{\sigma^*,\tau} (\alpha)) = 0 \iff 
\lim_{\alpha\to 1^-} \; (1-\alpha)^{-(\sens'' +2)} \Delta (\alpha)= 0 
    \]
for all $\sens''$. 
We conclude that $g_0=\ldots = g_{\sens'+1}=0$ and $g_{\sens'+2}\neq 0$ in~\eqref{eq:bound_gi}. Thus, by~\eqref{Eq1Th1} and \Cref{prop_threshold_sto}, it follows that $v_i^{\sigma^*, \tau^*} (\alpha) - v_i^{\sigma^*,\tau} (\alpha) < 0$ for $1-\frac{2^{\sens'+1}}{nW (2M)^{2n-1} {2n-1 \choose \sens'+3}} < \alpha < 1$. Since the function $j \mapsto \frac{2^{j-1}}{nW (2M)^{2n-1} {2n-1 \choose j+1}}$ is convex and attains its minimum at $j = \lfloor \frac{2}{3} n - 1 \rfloor$, and $\sens' \leq \sens$, we conclude that $v_i^{\sigma^*, \tau^*} (\alpha) - v_i^{\sigma^*,\tau} (\alpha) < 0$ for $1-\frac{2^{\min \{ \sens+2 , \lfloor \frac{2}{3} n - 1 \rfloor \} - 1}}{nW (2M)^{2n-1} {2n-1 \choose \min \{ \sens+2 , \lfloor \frac{2}{3} n - 1 \rfloor \} + 1}} < \alpha < 1$. Thus, in particular we have $v_i^{\sigma^*, \tau^*} (\alpha') - v_i^{\sigma^*,\tau} (\alpha') < 0$, contradicting the fact that $(\sigma^*,\tau^*)$ is a pair of discount optimal strategies for $\alpha'$.

Now, if we assume that there exist a strategy $\sigma$ and a state $i$ such that~\eqref{Eq2Th1} is satisfied, using similar arguments as above we arrive at a contradiction. 

This shows that $(\sigma^*,\tau^*)$ is a pair of $\sens $-sensitive discount optimal strategies.
    \end{proof}
    
\begin{remark}
Setting $M=1$ in our bounds for the stochastic case do not recover our results for the deterministic case (\Cref{SectionDeterm}). This is because, in the case of deterministic transitions, we can exploit the ``path then circuit structure" of the discounted value functions, as highlighted in \Cref{SectionDeterm}.
\end{remark}    
\subsection{Proof of \Cref{th-dubickas2}}
\begin{proof}
The proof follows similar arguments to the ones in the proof of \Cref{th-dubickas}, so given $\epsilon > 0$, let $D_\epsilon$ be the constant provided by \Cref{th:mahler}.

 Let $z$ be any real root of the polynomial $\Delta$ different from  $1$, $d$ be its degree and $P =\sum^{d}_{k=0} c_k x^k$ be its minimal polynomial. Since $z$ is a root of $\Delta$, it follows that $P$ divides $\Delta$, and so we have $d \leq 2n -1$ and $M(P) \leq M(\Delta)$. 
% \todo{RK: this sentence is critical and needs to be checked.}
Besides, note that by \Cref{prop:bound on Delta - stochastic} and Landau's bound~\cite{Landau1905}, we have $M(\Delta)\leq 2 n W M^{2n-1} \sqrt{{2(2n - 1) \choose 2n -1}}$.
% \todo{RK: please check this inequality}

If $d > D_\epsilon$, then~\eqref{dubickas} holds, and so we have
\begin{align*}
|z - 1 | & > e^{-(\pi/4 + \epsilon) \sqrt{d \log d  \log M(P)}} \geq e^{-(\pi/4 + \epsilon) \sqrt{(2n-1) \log (2n-1) \log M(\Delta)}} \\
& \geq e^{-(\pi/4 + \epsilon) \sqrt{(2n-1) \log (2n-1)  \log \left( 2 n W M^{2n-1} \sqrt{{2(2n - 1) \choose 2n -1}} \right) }} \; . 
\end{align*}

Assume now that $d \leq D_\epsilon$. By Lemma~\ref{lemmaTechnical2} it follows that $|z -1 | \geq \frac{1}{2 H(P) {d+1\choose \lceil \frac{d}{2}\rceil}}$. Then, since $H(P) \leq 2^d M(P)$, we have
\[
|z -1 | \geq \frac{1}{2^{d+1} M(P) {d+1\choose \lceil \frac{d}{2}\rceil }}
\geq \frac{1}{2^{D_\epsilon +1} M(\Delta) {D_\epsilon +1\choose \lceil \frac{D_\epsilon}{2}\rceil }}
\geq \frac{1}{2^{D_\epsilon +1} {D_\epsilon +1\choose \lceil \frac{D_\epsilon}{2}\rceil } 2 n W M^{2n-1} \sqrt{{2(2n - 1) \choose 2n -1}} } \enspace .
\]

Setting $a_\epsilon = \log\left(2^{D_\epsilon +1} {D_\epsilon +1\choose \lceil \frac{D_\epsilon}{2}\rceil } \right)$,
we conclude that $\Delta$ has no zeros in the interval  $]\alphama , 1[$. %and so the difference $v_i^{\sigma, \tau} (\alpha) - v_i^{\sigma',\tau'} (\alpha)$ has constant sign in this interval for any pairs of strategies $(\sigma, \tau)$ and $(\sigma', \tau')$. 
This shows that $\alphabw \leq \alphama$.
\end{proof}

\end{document}